\documentclass[journal]{IEEEtran}
\usepackage{balance}
\usepackage{graphicx}
\usepackage{bm}
\usepackage{amsmath}
\usepackage{amsfonts}
\usepackage{amsthm}
\usepackage{amssymb}
\usepackage{comment}
\usepackage{color}
\usepackage{mathrsfs}
\usepackage{upgreek}
\usepackage[utf8]{inputenc}
\usepackage{tikz}
\usetikzlibrary{arrows, calc, decorations.markings, positioning}
\usepackage{cite}
\usepackage{booktabs}
\usepackage{multirow}
\usepackage{amsmath}
\usepackage{amsthm}
\usepackage{algorithm}
\usepackage{algorithmic}
\usepackage{midfloat}
\usepackage{tabularx}
\usepackage{threeparttable}
\ifCLASSOPTIONcompsoc
 \usepackage[caption=false,font=normalsize,labelfont=sf,textfont=sf]{subfig}
\else
 \usepackage[caption=false,font=footnotesize]{subfig}
\fi
\newtheorem{thm}{Theorem}
\newtheorem{assumption}{Assumption}
\newtheorem{definition}{Definition}
\newtheorem{lemma}{Lemma}
\newtheorem{corollary}{Corollary}

\newtheorem{proposition}{Proposition}

\newtheorem{conjecture}{Conjecture}

\usepackage{hyperref}					
\hypersetup{colorlinks,
	linkcolor=blue,%
	anchorcolor=blue,
    citecolor=blue}

\begin{document}
\title{\huge{CP-OFDM Achieves the Lowest Average Ranging Sidelobe Under QAM/PSK Constellations}
}
\author{
	{
	Fan Liu,~\IEEEmembership{Senior Member,~IEEE}, Ying Zhang,~\IEEEmembership{Graduate Student Member,~IEEE}, \\Yifeng Xiong,~\IEEEmembership{Member,~IEEE}, Shuangyang Li,~\IEEEmembership{Member,~IEEE}, Weijie Yuan,~\IEEEmembership{Senior Member,~IEEE},\\ Feifei Gao,~\IEEEmembership{Fellow,~IEEE}, Shi Jin,~\IEEEmembership{Fellow,~IEEE}, and~Giuseppe Caire,~\IEEEmembership{Fellow,~IEEE}
	} 
\thanks{F. Liu and S. Jin are with the National Mobile Communications Research Laboratory,
Southeast University, Nanjing 210096, China (e-mail: fan.liu@seu.edu.cn, jinshi@seu.edu.cn).}
\thanks{Y. Zhang, and W. Yuan are with the School of Automation and Intelligent Manufacturing (AIM), Southern University of Science and Technology, Shenzhen 518055, China. (e-mail: yuanwj@sustech.edu.cn, yingzhang98@ieee.org).}
\thanks{Y. Xiong is with the School of Information and Electronic Engineering,
Beijing University of Posts and Telecommunications, Beijing 100876, China. (e-mail: yifengxiong@bupt.edu.cn)}
\thanks{S. Li and G. Caire are with the Chair of Communications and Information Theory, Technical University of Berlin, 10623 Berlin, Germany (e-mail: shuangyang.li@tu-berlin.de, caire@tu-berlin.de).}
\thanks{F. Gao is with the Department of Automation, Tsinghua University,
Beijing 100084, China (e-mail: feifeigao@ieee.org).}
}
\maketitle

\begin{abstract}
This paper aims to answer a fundamental question in the area of Integrated Sensing and Communications (ISAC): \textit{What is the optimal communication-centric ISAC waveform for ranging?} Towards that end, we first established a generic framework to analyze the sensing performance of communication-centric ISAC waveforms built upon orthonormal signaling bases and random data symbols. Then, we evaluated their ranging performance by adopting both the periodic and aperiodic auto-correlation functions (P-ACF and A-ACF), and defined the expectation of the integrated sidelobe level (EISL) as a sensing performance metric. On top of that, we proved that among all communication waveforms with cyclic prefix (CP), the orthogonal frequency division multiplexing (OFDM) modulation is the only globally optimal waveform that achieves the lowest ranging sidelobe for quadrature amplitude modulation (QAM) and phase shift keying (PSK) constellations, in terms of both the EISL and the sidelobe level at each individual lag of the P-ACF. As a step forward, we proved that among all communication waveforms without CP, OFDM is a locally optimal waveform for QAM/PSK in the sense that it achieves a local minimum of the EISL of the A-ACF. Finally, we demonstrated by numerical results that under QAM/PSK constellations, there is no other orthogonal communication-centric waveform that achieves a lower ranging sidelobe level than that of the OFDM, in terms of both P-ACF and A-ACF cases.
\end{abstract}

\section{Introduction}
\IEEEPARstart{E}{nvisioned} as a transformative paradigm, the sixth generation (6G) of wireless networks is set to drive forward emerging applications such as autonomous vehicles, smart factories, digital twins, and low-altitude economy \cite{Chafii2023CST,saad2019vision}. This new era of connectivity will extend beyond traditional communication roles, ushering in the ISAC technology \cite{9737357}. Notably, the International Telecommunication Union (ITU) has recently endorsed the global 6G vision, highlighting ISAC as one of its six primary usage scenarios \cite{ITU2023}.

The core concept of ISAC in 6G networks is to leverage wireless resources such as time, frequency, beam, and power across both sensing and communication functionalities using a unified hardware platform \cite{9606831}. The main challenge in ISAC lies in developing a dual-functional waveform that can effectively handle both target information acquisition and communication information delivery over the ISAC channel \cite{10334037,10292797}, which generally follows three design philosophies: sensing-centric, communication-centric, and joint designs \cite{9737357}. While the sensing-centric methodology aims at embedding communication information into existing radar sensing waveforms, e.g., chirp signals, its communication-centric counterpart seeks to implement sensing over standardized communication signaling schemes. In contrast to those, the joint design approach creates novel ISAC waveforms from scratch, aiming to balance and optimize the tradeoff between sensing and communication \cite{liu2021CRB}. 

While each of the aforementioned designs has its own applicable scenarios, the communication-centric approach is anticipated to be more favorable in future 6G ISAC networks due to its low implementation complexity \cite{10012421,9921271}. This approach allows for the direct use of a communication waveform for sensing, eliminating the need for waveform reshaping. Unlike channel estimation, which only employs known pilot symbols, typical monostatic or cooperative bi-static ISAC systems benefit from the full knowledge of the emitted waveform being shared between the ISAC transmitter (Tx) and the sensing receiver (Rx) \cite{10188491}. This enables the use of both pilot and data symbols for sensing, thereby enhancing range and Doppler resolutions by fully exploiting time-frequency resources. However, to convey useful information, the communication data symbols have to be \textit{random}, which may degrade the sensing performance. This has been recently identified as a fundamental deterministic-random tradeoff (DRT) in ISAC systems \cite{10471902,10147248,zhang2023input}. Consequently, it is essential to seek for optimal communication-centric waveforms that minimize the loss in sensing performance.

Classical communication waveforms convey data symbols using a well-designed orthonormal basis. In its simplest form, such a basis may consist of time-shifted unit impulse functions, corresponding to single-carrier (SC) waveforms \cite{8114253}. In contrast, OFDM waveforms modulate frequency-domain symbols using multiple sinusoidal subcarriers centered at different frequencies, leveraging the inverse discrete Fourier transform (IDFT) matrix as the signaling basis \cite{6923528,10012421}. Additionally, code-division multiple access (CDMA) schemes use pseudo-random codes, such as Walsh codes constructed by Hadamard matrices, to carry information symbols \cite{771348}. To address the time-frequency doubly selective effect of the high-mobility channels, orthogonal time-frequency space (OTFS) modulation has been proposed as a potential 6G waveform, modulating symbols in the delay-Doppler domain through the inverse symplectic finite Fourier transform (ISFFT) \cite{7925924,9508932}. More recently, a novel affine frequency division multiplexing (AFDM) waveform was conceived for high-mobility communications, which places symbols in the affine Fourier transform (AFT) domains with orthogonal chirp signals as signaling basis \cite{10087310,10439996}. Against this background, an important, yet unresolved question is: \textit{What is the optimal communication-centric ISAC waveform under random signaling?}

A substantial body of work has focused on the analysis and design of communication-centric ISAC waveforms. Pioneered by Sturm and Wiesbeck, the feasibility of using OFDM waveforms to measure delay and Doppler parameters of radar targets has been investigated in \cite{sturm2011waveform} for single-antenna systems, which has been recently generalized to the multi-antenna counterparts \cite{10175076,9149408,9682124}. The authors in \cite{9359665} proposed a code-division OFDM (CD-OFDM) waveform for ISAC, combining CDMA and OFDM techniques. The sensing performance of the cyclic prefixed single-carrier (CP-SC) waveform was examined in \cite{9005192}. Moreover, the study in \cite{chen2023fdss} demonstrated that OFDM offers superior ranging performance compared to discrete Fourier transform spread OFDM (DFT-s-OFDM), which can be considered as a specific OFDM signaling scheme with SC characteristics. A recent debate has emerged on whether OTFS outperforms OFDM in terms of the sensing performance. The authors in \cite{9109735} investigated the range and velocity estimation errors of both waveforms, finding that OFDM performs slightly better than OTFS. In \cite{10638525}, the ambiguity functions of both waveforms were illustrated, suggesting that OTFS produces lower sidelobes in both delay and Doppler domains. However, the comparison in \cite{10638525} was biased, as it utilized random QPSK symbols for OFDM but deterministic symbols for OTFS. More recently, a dual-domain ISAC signaling scheme that integrates both OFDM and OTFS was proposed in \cite{10264814}. Although existing works have significantly advanced the optimization of communication-centric ISAC waveforms, none has addressed the underlying question in a rigorous manner.

In this paper, we attempt to partially answer this open question from the perspective of target range estimation under a monostatic ISAC setup, where a Tx emits a communication-centric ISAC signal generated by modulating random information symbols over an orthonormal signaling basis. The signal is received by communication users, while being reflected back from distant targets to a sensing Rx collocated with the ISAC Tx. As a consequence, the ISAC signal is fully known to the sensing Rx despite its randomness. The ranging performance is then evaluated by the auto-correlation function (ACF) of the ISAC signal, under both periodic and aperiodic convolutions, corresponding to the matched-filtering operation for signals with and without CP, respectively. For clarity, we summarize our main contributions as follows:
\begin{itemize}
    \item We developed a generic framework to analyze the sensing performance of communication-centric ISAC waveforms built upon orthonormal signaling basis and random data symbols. Specifically, we analyzed both the P-ACF and A-ACF of random ISAC signals, and defined the expectation of the integrated sidelobe level (EISL) as a performance metric for ranging.
    \item We derived closed-form expressions of the sidelobe levels of both the P-ACF and A-ACF for random ISAC signals, under various types of random communication symbols including sub-Gaussian (e.g., QAM and PSK) and super-Gaussian (e.g., specific Amplitude and Phase-Shift Keying (APSK), and index modulation) constellations.
    \item We proved that among all communication-centric ISAC waveforms with CP, OFDM is the \textit{only} globally optimal waveform that achieves the lowest ranging sidelobe level for standard QAM/PSK constellations, in terms of both EISL and each individual sidelobe index of its P-ACF. As a direct corollary, we also proved that the CP-SC waveform achieves the lowest Doppler sidelobe level.
    \item We proved that among all communication-centric ISAC waveforms without CP, OFDM is locally optimal in the sense that it achieves a local minimum of the EISL of the A-ACF for QAM/PSK constellations. We conjecture that OFDM is also the global EISL minimizer for signals without CP under sub-Gaussian constellations.
\end{itemize}

The remainder of this paper is organized as follows. Sec. \ref{sec_2} introduces the system model of the ISAC system and the corresponding performance metrics. Sec. \ref{sec_3} and \ref{sec_4} elaborate on the P-ACF and A-ACF for random ISAC signals and the corresponding optimal signaling strategies, respectively. Sec. \ref{sec_5} provides simulation results to validate the theoretical analysis of the paper. Finally, Sec. \ref{sec_6} concludes the paper.
\\\indent {\emph{Notations}}: Matrices are denoted by bold uppercase letters (e.g., $\mathbf{U}$), vectors are represented by bold lowercase letters (e.g., $\mathbf{x}$), and scalars are denoted by normal font (e.g., $N$); The $n$th entry of a vector $\mathbf{s}$, and the $(m,n)$-th entry of a matrix $\mathbf{A}$ are denoted as $s_n$ and $a_{m.n}$, or $\left[\mathbf{s}\right]_n$ and $\left[\mathbf{A}\right]_{m,n}$, respectively; $\otimes$, $\odot$ and $\operatorname{vec}\left(\cdot\right)$ denote the Kronecker product, the Hadamard product, and the vectorization in terms of the columns of input matrices; $\left(\cdot\right)^T$, $\left(\cdot\right)^H$, and $\left(\cdot\right)^\ast$ stand for transpose, Hermitian transpose, and complex conjugate of the matrices; The entry-wise square of a matrix $\mathbf{A}$ is denoted as $ \mathbf{X}\odot\mathbf{X}^\ast \triangleq \left|\mathbf{X}\right|^2$; $\mathbf{A}\cdot\mathbf{B}$ represents the row-wise Kronecker product between matrices $\mathbf{A}$ and $\mathbf{B}$; $\operatorname{Re}\left(\cdot\right)$ and $\operatorname{Im}\left(\cdot\right)$ denote the real and imaginary parts of the argument; The $\ell_p$ norm and Frobenius norm are written as $\left\| \cdot\right\|_p$ and $\left\| \cdot\right\|_F$; $\mathbb{E}(\cdot)$ and $\text{Var}(\cdot)$ represent the expectation and variance; $\circledast$ denotes the circular convolution; The notation ${\rm diag}(\mathbf{a})$ denotes the diagonal matrix obtained by placing the entries of $\mathbf{a}$ on its main diagonal, while ${\rm ddiag}(\mathbf{A})$ denotes the vector obtained by extracting the main diagonal entries from $\mathbf{A}$; The symbol $\delta_{m,n}$ denotes the Kronecker delta function given by $$\delta_{m,n}=\left\{\begin{aligned}0,~~&\hbox{$m\neq n$;}\\ 1,~~&\hbox{$m=n.$}\end{aligned}\right.$$

\begin{table}
\renewcommand{\arraystretch}{1.3}

\caption{List of Acronyms}
\label{table_example}
\centering
\begin{tabular}{@{}|l|r|@{}}
\hline
ISAC & Integrated Sensing and Communications \\
ACF & Auto-correlation Function \\
P-ACF & Periodic ACF \\
A-ACF & Aperiodic ACF \\
CP & Cyclic Prefix \\
DFT & Discrete Fourier Transform \\
IDFT & Inverse Discrete Fourier Transform \\
QAM & Quadrature Amplitude Modulation \\
PSK & Phase Shift Keying \\
APSK & Amplitude Phase Shift Keying \\
ISL & Integrated Sidelobe Level \\
EISL & Expectation of Integrated Sidelobe Level \\
SC & Single-Carrier \\
OFDM & Orthogonal Frequency Divsion Multiplexing \\
CDMA & Code-Division Multiple Access \\
OTFS & Orthogonal Time-Frequency Space \\
AFDM & Affine Frequency Division Multiplexing \\
\hline
\end{tabular}

\end{table}

\section{System Model}\label{sec_2}
\subsection{ISAC Signal Model}
We consider a monostatic ISAC system as shown in Fig. \ref{ISAC_Scenario}. The ISAC Tx emits an ISAC signal modulated with random communication symbols, which is received at a communication Rx, and is simultaneously reflected back to the sensing Rx by one or more targets at different ranges. The sensing Rx, which is collocated with the ISAC Tx, performs matched-filtering to estimate the delay parameters of targets by using the known random ISAC signal.

Let $\mathbf{s} = \left[s_1, s_2,\ldots,s_{N}\right]^T\in \mathbb{C}^{N\times 1}$ be $N$ communication symbols to be transmitted. We assume that each symbol is randomly drawn from a complex alphabet $\mathcal{S}$ in an i.i.d. manner, which is also known as a \textit{constellation}. Without loss of generality, we adopt the following assumptions for the considered constellations.
\begin{assumption}[Unit Power]
    We focus on constellations with a unit power, namely, 
    \begin{equation}
        \mathbb{E}(|s|^2) = 1,\quad \forall s\in\mathcal{S}.
    \end{equation}
\end{assumption}
\noindent Assumption 1 normalizes the power of the constellations such that their sensing and communication performance could be fairly compared.
\begin{assumption}[Rotational Symmetry]
    The expectation and pseudo variance of the constellation are zero, namely
    \begin{equation}
        \mathbb{E}(s) = 0,\quad\mathbb{E}(s^2) = 0, \quad \forall s\in\mathcal{S}.
    \end{equation}
\end{assumption}
\noindent We remark that most of the commonly employed constellations meet the criterion in Assumption 2, including all the PSK and QAM constellations except for BPSK and 8-QAM. Nevertheless, we will show in later sections that our results are indeed applicable to BPSK, which makes 8-QAM the only outlier of the proposed framework. 

Let us further define
\begin{equation}
    \mu_4 \triangleq \frac{\mathbb{E}\left\{|s-\mathbb{E}(s)|^4\right\}}{\mathbb{E}\left\{|s-\mathbb{E}(s)|^2\right\}^2} =  \mathbb{E}(|s|^4),
\end{equation}
which is known as the \textit{kurtosis} of the constellation, and is equivalent to its 4th-order moment if the constellation has zero mean and unit power. Note that $\mu_4 \ge 1$ since the 4th-moment is greater than the square of the 2nd-moment. 

It is worth highlighting that the standard complex Gaussian distribution also satisfies the above criteria, which makes it an adequate constellation. Indeed, it is well-known that the Gaussian distributed constellation achieves the capacity of a Gaussian channel. As we shall see later, the circularly symmetric complex Gaussian constellation with unit variance, which has a kurtosis of 2, also serves as an important baseline in terms of the sensing performance of the ISAC signal. Accordingly, it would be useful to define the following two types of constellations.
\begin{definition}[Sub-Gaussian Constellation]
    A sub-Gaussian constellation is a constellation with kurtosis less than 2, subject to Assumptions 1 and 2.
\end{definition}
\begin{definition}[Super-Gaussian Constellation]
    A super-Gaussian constellation is a constellation with kurtosis greater than 2, subject to Assumptions 1 and 2.
\end{definition}

\begin{figure}[!t]
	\centering
	\includegraphics[scale=0.25]{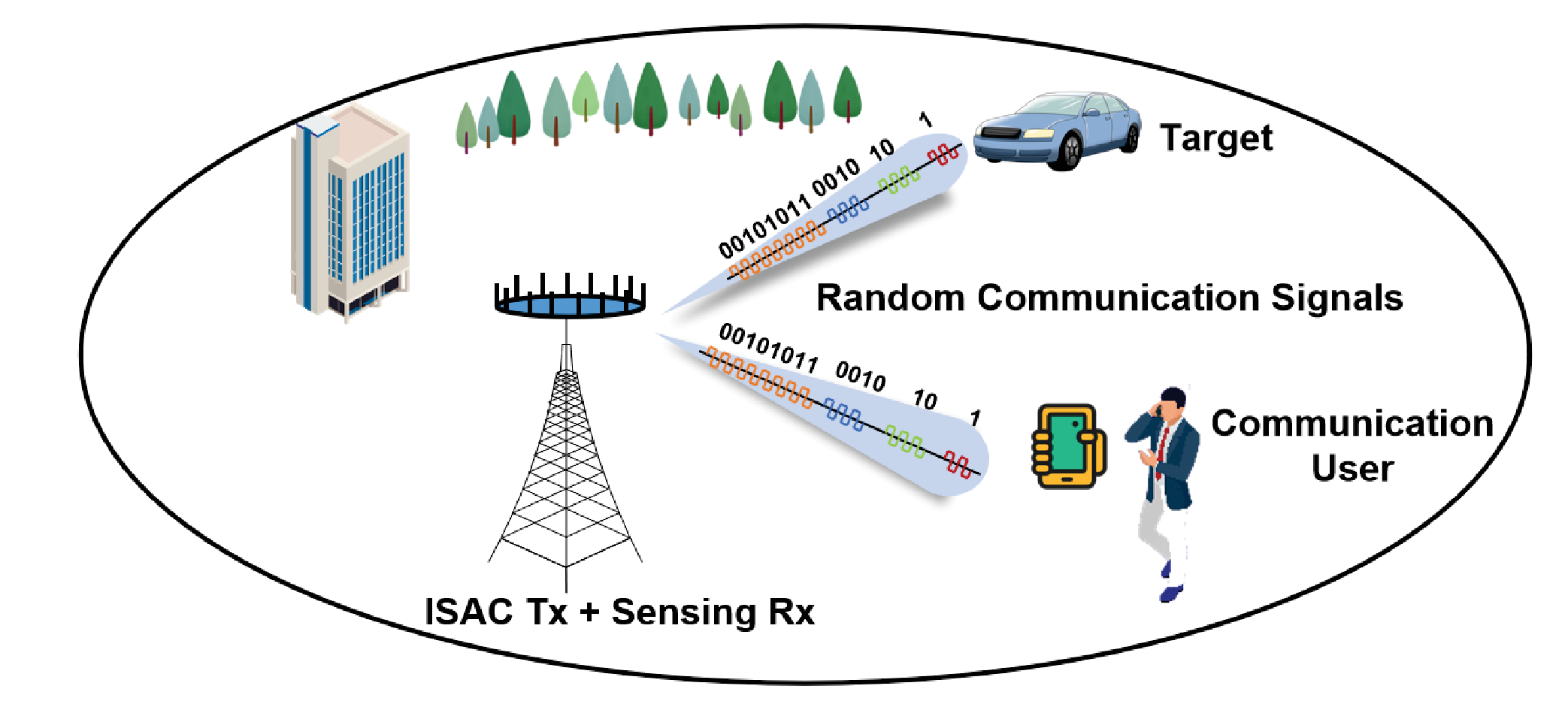}
	\caption{The ISAC transmission scenario using communication-centric waveform.}
        \label{ISAC_Scenario}
\end{figure} 

All the QAM and PSK constellations are sub-Gaussian. In particular, the kurtosis is equal to 1 for all PSK constellations, and is between 1 and 2 for all QAM constellations. For clarity, we show the kurtosis values of typical QAM and PSK constellations in TABLE. \ref{tab: kurtosis}. There are certain types of super-Gaussian constellations with drastically varying amplitudes. Alternatively, one may also generate super-Gaussian constellations by applying geometric or probabilistic constellation shaping techniques. Index modulation (IM) \cite{8315127,8004416} is one of such techniques that are capable of producing super-Gaussian constellations. To elaborate, let us consider a constellation having a kurtosis of $\mu_4$. By applying the IM, a probability mass would be placed at the origin of the I/Q plane, while leaving the rest of the constellation unaltered. Assuming that the probability of transmitting the origin is $p_0$, the resulting kurtosis would then become
\begin{equation}\label{im_kurtosis}
\widetilde{\mu}_4 =\frac{(1-p_0)\mathbb{E}\left\{|s-\mathbb{E}(s)|^4\right\}}{\mathbb{E}\left\{(1-p_0)|s-\mathbb{E}(s)|^2\right\}^2}=\frac{\mu_4}{1-p_0} \ge \frac{1}{1-p_0}.
\end{equation}
Observe that having an origin-transmitting probability of $p_0>\frac{1}{2}$ would effectively transform any constellation into a super-Gaussian one. Super-Gaussian constellations can be useful in scenarios where energy efficiency is of priority or non-coherent communication schemes are considered \cite{923716,1532206,1532207}.

\begin{table}[!t]
\caption{Kurtosis values of typical sub-Gaussian constellations}
\label{tab: kurtosis}
\begin{tabular}{l|c|c|c|c}
\hline
\textbf{Constellation} & PSK     & 16-QAM  & 64-QAM   & 128-QAM  \\ \hline
\textbf{Kurtosis}      & 1       & 1.32    & 1.381   & 1.3427   \\ \hline
\textbf{Constellation} & 256-QAM & 512-QAM & 1024-QAM & 2048-QAM \\ \hline
\textbf{Kurtosis}      & 1.3953  & 1.3506  & 1.3988   & 1.3525   \\ \hline
\end{tabular}
\end{table}
In typical communication systems, we shall modulate $N$ symbols over an orthonormal basis on the time domain, which may be defined as a unitary matrix $\mathbf{U} = \left[\mathbf{u}_1,\mathbf{u}_2,\ldots,\mathbf{u}_{N}\right]\in \mathbb{U}\left(N\right)$, where $\mathbb{U}\left(N\right)$ denotes the unitary group of degree $N$. Consequently, the discrete time-domain signal can be expressed as
\begin{equation}
    \mathbf{x} = \mathbf{U}\mathbf{s} = \sum\limits_{n = 1}^{N} {s_n\mathbf{u}_n}.
\end{equation}
The above generic model may represent most of the communication signaling schemes, where we show some examples below:
\begin{itemize}
    \item \textbf{SC:} $\mathbf{U} = \mathbf{I}_N$. In this case, we are simply transmitting symbols consecutively on the time domain, where the orthonormal basis is composed of nothing but unit impulse functions, yielding an SC signal.
    \item \textbf{OFDM:} $\mathbf{U} = \mathbf{F}_N^H$, where $\mathbf{F}_N$ is the normalized discrete Fourier transform (DFT) matrix of size $N$. The symbols are placed in the frequency domain, which makes $\mathbf{x}$ an OFDM signal with $N$ subcarriers. If $L$ OFDM symbols are transmitted with $M = N/L$ subcarriers, then $\mathbf{U} = \mathbf{I}_L \otimes \mathbf{F}_{M}^H$.
    \item \textbf{CDMA:} $\mathbf{U} =  \mathbf{C}_N$, where $\mathbf{C}_N$ is the Hadamard matrix of size $N$. The symbols are then placed in the code domain, making $\mathbf{x}$ a CDMA signal that has been extensively used in CDMA2000 \cite{995853}.
    \item \textbf{OTFS:} $\mathbf{U} =\mathbf{F}_{M}^H \otimes \mathbf{I}_L$. In this case, $\mathbf{s}$ is placed in the delay-Doppler domain, where the number of occupied time slots and subcarriers are $M$ and $L$, respectively \cite{10463758}.
    \item \textbf{AFDM:} $\mathbf{U} = \mathbf{\Lambda}_{c_1}^H\mathbf{F}_N^H\mathbf{\Lambda}_{c_2}^H$  \cite{10087310,10439996}, where $\mathbf{\Lambda}_{c} = \operatorname{Diag}(1,e^{-j2\pi c1^2},\ldots,e^{-j2\pi cN^2})$. In this case, $\mathbf{U}$ is the inverse discrete affine Fourier transform (IDAFT) matrix, and the symbols are placed in the AFT domain.
\end{itemize}

We highlight here that in many communication signaling schemes, such as OFDM, CP-SC (also known as single-carrier frequency-domain equalization (SC-FDE)) \cite{995852}, OTFS, and AFDM, the addition of a cyclic prefix (CP) is necessary, which eliminates the inter-symbol interference (ISI) and reduces the computational complexity by processing the received signal in the frequency/delay-Doppler domain. Nevertheless, adding CPs to the signal may not always be a convention in radar sensing systems. This is because for long-range detection tasks, the targets may be located far beyond the coverage of a CP, in which case a sufficiently small duty cycle ($<10\%$) is required, rendering the ISAC signal as a zero-padding waveform in contrast to its CP'ed counterpart.

Towards that end, we present in this paper a thorough analysis on both cases, i.e., ISAC signaling with and without CP, which correspond to periodic and linear convolution processing of the matched-filter at the sensing Rx, respectively. Without loss of the generality, whenever a CP is added, it is assumed to be larger than the maximum delay of the communication paths and sensing targets.

\subsection{Communication System Model}
Before delving into the technical details, let us first briefly examine the communication performance of the ISAC signal under two typical communication channels: The additive white Gaussian noise (AWGN) channel and linear time-invariant (LTI) multi-path channel, where we assume a CP is added to the signal for the purpose stated above. In the AWGN case, the received signal after CP removal reads
\begin{equation}
    \mathbf{y}_c = \sqrt{\rho}\mathbf{x} + \mathbf{z} = \sqrt{\rho}\mathbf{U}\mathbf{s} + \mathbf{z},
\end{equation}
where $\rho$ is the signal-to-noise ratio (SNR), and $\mathbf{z}\sim \mathcal{CN}(\mathbf{0},\mathbf{I}_N)$ denotes the WGN with zero mean and unit variance. It then holds immediately that
\begin{equation}
    I(\mathbf{y}_c;\mathbf{x}) = I(\sqrt{\rho}\mathbf{U}\mathbf{s} + \mathbf{z};\mathbf{U}\mathbf{s}) = I(\mathbf{y}_c;\mathbf{s}).
\end{equation}
That is, unitary transform keeps the input-output mutual information (MI) unchanged in an AWGN channel, thereby preserving the communication achievable rate. More precisely, in an AWGN channel, the communication symbols may be modulated over arbitrary orthonormal basis, without performance loss in the communication rate.

Let us now turn our focus to the LTI multi-path channel with WGN, which outputs the following signal at the communication Rx after removing the CP:
\begin{equation}
    \mathbf{y}_c = \sqrt{\rho}\mathbf{H}_c\mathbf{x} + \mathbf{z} = \sqrt{\rho}\mathbf{H}_c\mathbf{U}\mathbf{s} + \mathbf{z},
\end{equation}
where $\mathbf{H}_c$ is a circulant matrix, with each row containing circularly ordered entries, namely the complex channel gains at multiple delayed paths. In such a channel, the MI is no longer invariant under unitary transform \cite{5744135}. It then follows that OFDM signaling achieves the optimal MI in LTI multi-path channels, since $\mathbf{H}_c$ is diagonalized by $\mathbf{U}$, and an optimal power allocation strategy can thereby be employed. One would then wonder whether the OFDM is still optimal for sensing, especially under random ISAC signaling. This represents a question receiving much attention in the area of ISAC: What is the optimal communication-centric ISAC waveform? In later sections, we attempt to answer this question by proving that the OFDM is the optimal waveform that achieves the lowest ranging sidelobe under QAM/PSK alphabets.



\subsection{Sensing Performance Metric}
The auto-correlation function (ACF) of the signal is an important performance indicator for ranging, in particular for the matched filtering process at the sensing Rx. The ACF may be defined as the aperiodic or periodic self-convolution of the signal, depending on whether a CP is added to the signal.
\subsubsection{Aperiodic ACF (A-ACF)}
\begin{equation}
    r_k = \mathbf{x}^H\mathbf{J}_k\mathbf{x} = r_{-k}^*, \quad k = 0,1,\ldots, N-1,
\end{equation}
where $\mathbf{J}_k$ is the $k$th shift matrix in the form of
\begin{equation}
    \mathbf{J}_k = \left[ {\begin{array}{*{20}{c}}
  {\mathbf{0}}&{{{\mathbf{I}}_{N - k}}} \\ 
  {\mathbf{0}}&{\mathbf{0}} 
\end{array}} \right].
\end{equation}
Given the symmetry of the ACF, we have
\begin{equation}
    \mathbf{J}_{-k} = \mathbf{J}_k^T= \left[ {\begin{array}{*{20}{c}}
  {\mathbf{0}}&{\mathbf{0}} \\ 
  {{{\mathbf{I}}_{N - k}}}&{\mathbf{0}} 
\end{array}} \right]. 
\end{equation}
\subsubsection{Periodic ACF (P-ACF)}
\begin{equation}
    \tilde{r}_k = \mathbf{x}^H\tilde{\mathbf{J}}_k\mathbf{x} = \tilde{r}_{N-k}^*, \quad k = 0,1,\ldots, N-1,
\end{equation}
where $\tilde{\mathbf{J}}_k$ is defined as the $k$th periodic shift matrix \cite{4838816}, given as
\begin{equation}
    \tilde{\mathbf{J}}_k = \left[ {\begin{array}{*{20}{c}}
  {\mathbf{0}}&{{{\mathbf{I}}_{N - k}}} \\ 
  {\mathbf{I}_k}&{\mathbf{0}} 
\end{array}} \right],
\end{equation}
and
\begin{equation}
    \tilde{\mathbf{J}}_{N-k} = \tilde{\mathbf{J}}_k^T= \left[ {\begin{array}{*{20}{c}}
  {\mathbf{0}}&{{\mathbf{I}_k}} \\ 
  {{{\mathbf{I}}_{N - k}}}&{\mathbf{0}} 
\end{array}} \right]. 
\end{equation}

In both cases, one may be concerned by the sidelobe level of the ACF, which plays a critical role in multi-target detection problems. Let us take the A-ACF as an example. The sidelobe of $r_k$ is defined as
\begin{equation}
    |r_k|^2 = |\mathbf{x}^H\mathbf{J}_k\mathbf{x}|^2 = |r_{-k}|^2, \quad k = 1,\ldots, N-1, 
\end{equation}
where $|r_0|^2 = |\mathbf{x}^H\mathbf{x}|^2 = \left\|\mathbf{x}\right\|_2^4$ is the mainlobe of the ACF. Accordingly, the integrated sidelobe level (ISL) may be expressed as \cite{4749273,7093191}
\begin{equation}
    \text{ISL} = \sum\limits_{k = 1}^{N - 1}{|r_k|^2}.
\end{equation}
Due to the random nature of the ISAC signal $\mathbf{x}$, the ACF becomes a random function. Hence, a natural choice is to define the average of the sidelobe level as a performance metric\footnote{We note that, alongside the average sidelobe level, the maximum sidelobe level (MSL) may also serve as a performance indicator for random ACFs. Nonetheless, the MSL is a non-smooth function, typically presenting substantial difficulties in deriving analytical insights, particularly for random signaling. As such, we leave the investigation of MSL as a direction for future research.}. This can be represented by
\begin{equation}
    \mathbb{E}(|r_k|^2) = \mathbb{E}(|\mathbf{x}^H\mathbf{J}_k\mathbf{x}|^2) = \mathbb{E}(|\mathbf{s}^H\mathbf{U}^H\mathbf{J}_k\mathbf{U}\mathbf{s}|^2), \;\;\forall k,
\end{equation}
where the expectation is with respect to the random symbol vector $\mathbf{s}$, $\mathbb{E}(|r_0|^2) = \mathbb{E}(\left\|\mathbf{x}\right\|_2^4)$ denotes the average mainlobe, and $\mathbb{E}(|r_k|^2), k\ne 0$ characterizes the average sidelobe level at the index $k$. The expectation of the ISL (EISL) is therefore given by
\begin{equation}
    \text{EISL} = \sum\limits_{k = 1}^{N - 1}\mathbb{E}({|r_k|^2}) = \sum\limits_{k = 1}^{N - 1}\mathbb{E}(|\mathbf{s}^H\mathbf{U}^H\mathbf{J}_k\mathbf{U}\mathbf{s}|^2).
\end{equation}
Accordingly, seeking for the optimal signaling basis is equivalent to solving the following stochastic optimization problem:
\begin{equation}\label{stochastic_opt}
    \min_{\mathbf{U}\in\mathbb{U}(N)}\;\;\sum\limits_{k = 1}^{N - 1}\mathbb{E}(|\mathbf{s}^H\mathbf{U}^H\mathbf{J}_k\mathbf{U}\mathbf{s}|^2).    
\end{equation}

\textbf{Remark 1:} It is worth pointing out that if $\mathbf{s}$ is realized from a circularly-symmetric Gaussian codebook, namely, $s_n \sim \mathcal{CN}(0,1)$, then the average mainlobe and sidelobe levels of its ACFs will keep unchanged regardless of the choice of the signaling basis. This is due to the simple fact that the standard Gaussian distribution is unitary invariant. 

While attaining the analytical solution of \eqref{stochastic_opt} for an arbitrarily distributed constellation seems to be a highly challenging task, in what follows, we reveal that the OFDM waveform achieves the lowest ranging sidelobe for both QAM and PSK. That is to say, $\mathbf{U} = \mathbf{F}_N^H$ is a minimizer of the EISL for both the A-ACF and P-ACF when $\mathbf{s}$ is realized from QAM/PSK constellations. We will further prove that, OFDM is the only globally optimal waveform that minimizes not only the EISL but also the individual sidelobe level at each $k\ne 0$ in the P-ACF case.

\section{The P-ACF Case}\label{sec_3}
Although being similarly formulated, dealing with the P-ACF is generally a simpler task than its aperiodic counterpart. In this section, we present the main results of the P-ACF case. 

First, by leveraging the well-known Wiener-Khinchin theorem, we may simplify the P-ACF by exploiting the frequency-domain representation of the ISAC signal. By performing DFT to $\mathbf{x}$, we obtain
\begin{equation}
    \mathbf{F}_N\mathbf{x} = \mathbf{F}_N\mathbf{U}\mathbf{s} \triangleq \mathbf{V}^H\mathbf{s},
\end{equation}
where $\mathbf{V} = \mathbf{U}^H\mathbf{F}_N^H = \left[\mathbf{v}_1,\mathbf{v}_2,\ldots,\mathbf{v}_{N}\right]\in \mathbb{U}\left(N\right)$. Accordingly, the P-ACF vector $\tilde{\mathbf{r}} = \left[\tilde{r}_0,\tilde{r}_1,\ldots,\tilde{r}_{N-1}\right]^T$ may be recast in the following form of:
\begin{equation}
    \tilde{\mathbf{r}} = \sqrt{N}\mathbf{F}_N^H\left|\mathbf{F}_N\mathbf{x}\right|^2 = \sqrt{N}\mathbf{F}_N^H\left|\mathbf{V}^H\mathbf{s}\right|^2,
\end{equation}
yielding
\begin{equation}\label{P_ACF_expansion}
    \tilde{r}_k =\sum\limits_{n = 1}^{N}|\mathbf{v}_n^H\mathbf{s}|^2 e^{\frac{j2\pi k(n-1)}{N}},\quad k = 0,1,\ldots,N-1.
\end{equation}


\subsection{Main Results}
We are now ready to express the squared P-ACF in closed form.
\begin{proposition}\label{prop:pacf}
The average squared P-ACF is
\begin{equation}\label{squared_P-ACF}
    \mathbb{E}(|\tilde{r}_k|^2) = N^2\delta_{0,k}+N+(\mu_4-2)\left\|\mathbf{b}_k\right\|_2^2,
\end{equation}
where
\begin{equation}\label{bk_defined}
    \mathbf{b}_k = {\left[ {\sum\limits_{n = 1}^N {{{\left| {{v_{1,n}}} \right|}^2}{e^{\frac{{ - j2\pi k\left( {n - 1} \right)}}{N}}}, \ldots ,\sum\limits_{n = 1}^N {{{\left| {{v_{N,n}}} \right|}^2}{e^{\frac{{ - j2\pi k\left( {n - 1} \right)}}{N}}}} } } \right]^T},
\end{equation}
with $v_{m,n}$ being the $(m,n)$-th entry of $\mathbf{V}$.
\end{proposition}
\begin{proof}
    See Appendix \ref{prop_1_proof}.
\end{proof}

\textbf{Remark 2:} Proposition \ref{prop:pacf} suggests that the sidelobe of average squared P-ACF could be reduced by either increasing $N$ or decreasing $\mu_4$. Especially, the latter could be achieved by employing constellation shaping techniques.

\begin{corollary}\label{coro: mainlobe}
The average mainlobe level of the P-ACF may be expressed as
\begin{equation}\label{mainlobe_level}
    \mathbb{E}(|\tilde{r}_0|^2)= N^2 + (\mu_4 - 1)N.
\end{equation}
\end{corollary}
\begin{proof}
It can be readily shown that $\mathbf{b}_0 = \mathbf{1}_N$, and hence $\left\|\mathbf{b}_0\right\|_2^2 = N$, which immediately yields \eqref{mainlobe_level}.
An alternative proof is based on the fact that the unitary transform does not change the $l_2$ norm of the signal, leading to
\begin{align}
\mathbb{E}(|\tilde{r}_0|^2) \nonumber&= \mathbb{E}(\left\|\mathbf{U}\mathbf{s}\right\|_2^4) = \mathbb{E}(\left\|\mathbf{s}\right\|_2^4)= \sum\limits_{n = 1}^{N}\sum\limits_{m = 1}^{N}\mathbb{E}(|s_n|^2|s_m|^2) 
\\
\nonumber&= N\mathbb{E}(|s_n|^4) + \sum\limits_{n = 1}^{N}\sum\limits_{\begin{subarray}{l} 
  m = 1 \\ 
  m \ne n 
\end{subarray}}^N\mathbb{E}(|s_n|^2)\mathbb{E}(|s_m|^2)\\ &=N^2 + (\mu_4 - 1)N.
\end{align}
\end{proof}

\begin{proposition}\label{prop:eisl}
    The EISL of the P-ACF is
    \begin{equation}\label{EISL_closed_form}
        \sum\limits_{k = 1}^{N - 1}\mathbb{E}({|\tilde{r}_k|^2}) = N(N-1) + (\mu_4 - 2)N(\left\|\mathbf{F}_N\mathbf{U}\right\|_4^4-1),
    \end{equation}
    where $\left\|\cdot\right\|_4$ denotes the $\ell_4$ norm of the matrix.
\end{proposition}
\begin{proof}
    See Appendix \ref{prop_2_proof}.
\end{proof}

\subsection{Optimal Signaling Schemes}
We may now establish the optimality of the OFDM waveform in Theorem \ref{thm:pacf_eisl}.
\begin{thm}[Global Optimality of the OFDM for Ranging]\label{thm:pacf_eisl}
    OFDM is the only waveform that achieves the lowest EISL of the P-ACF for sub-Gaussian constellations, e.g., PSK and QAM.
\end{thm}
\begin{proof}
    To minimize the EISL in the $\mu_4 < 2$ case, it is obvious from Proposition \ref{prop:eisl} that one has to maximize $\left\|\mathbf{F}_N\mathbf{U}\right\|_4^4$ over the unitary group, i.e.,
    \begin{equation}\label{L4_max_easy}
        \max\limits_{\mathbf{U} \in \mathbb{U}(N)}\;\; \left\|\mathbf{F}_N\mathbf{U}\right\|_4^4.
    \end{equation}
    Note that the $\ell_4$-norm can be alternatively expressed as $\|\mathbf{F}_N\mathbf{U}\|_4^4=\sum_{n=1}^N\|\mathbf{x}_n\|_2^2$, where $[\mathbf{x}_n]_m = |[\mathbf{F}_N\mathbf{U}]_{n,m}|^2$, and $\mathbf{1}^T\mathbf{x}_n=1$ holds for all $i$ due to the unitarity of $\mathbf{F}_N\mathbf{U}$. Next, observe that the inequality
    \begin{equation}
      \|\mathbf{x}_n\|_1^2 \geq \|\mathbf{x}_n\|_2^2
    \end{equation}
    achieves its equality if and only if only one of the entries of $\mathbf{x}_n$ is $1$ while others are $0$. Together with the unitary constraint, this implies that the optimal $\mathbf{F}_N\mathbf{U}$ can only be complex permutation matrices. Consequently, the optimal signaling basis shall be expressed in the form of
    \begin{equation}\label{ofdm_opt}
        \mathbf{U}^\star = \mathbf{F}_N^H\mathbf{\Pi}\operatorname{Diag}({\boldsymbol{{\theta}}}),
    \end{equation}
    where $\mathbf{\Pi}\in \mathbb{R}^{N \times N}$ is any real permutation matrix, and ${\boldsymbol{{\theta}}} \in \mathbb{C}^{N}$ is a vector with unit-modulus entries. If $\mathbf{\Pi} = \mathbf{I}_N, {\boldsymbol{{\theta}}} = \mathbf{1}_N$, then $\mathbf{U}^\star$ represents the standard OFDM waveform. Otherwise, \eqref{ofdm_opt} simply results in an OFDM waveform with different initial phases over permuted subcarriers, completing the proof.
\end{proof}
From Theorem \ref{thm:pacf_eisl} it is obvious that for $\mu_4 < 2$ we have
\begin{equation}
    \sum\limits_{k = 1}^{N - 1}\mathbb{E}({|r_k|^2}) \ge (\mu_4 - 1)N(N-1),
\end{equation}
since $\left\|\mathbf{F}_N\mathbf{U}^\star\right\|_4^4 = N$. Furthermore, the following theorem provides a stronger result for the optimality of the OFDM.
\begin{thm}[OFDM Achieves the Lowest Average Sidelobe at Every Lag]\label{thm:pacf_everywhere}
OFDM is the only waveform that achieves the lowest average sidelobe level at every delay index $k$ of the P-ACF for sub-Gaussian constellations.
\end{thm}
\begin{proof}
    We first establish an upper-bound for $\left\|\mathbf{b}_k\right\|_2^2$, and then prove that it is achievable by the OFDM waveform. Since $\mathbb{E}(|\tilde{r}_k|^2)\geq 0$ for all $k$, for $k\ne 0$ we have
    \begin{equation}
        \mathbb{E}(|\tilde{r}_k|^2) = N+(\mu_4-2)\left\|\mathbf{b}_k\right\|_2^2 \ge 0.
    \end{equation}
    Given the fact that $\mu_4 \ge 1$, by letting $\mu_4=1$, we may then see that $\left\|\mathbf{b}_k\right\|^2 \le N$ at every sidelobe index $k$. Now suppose that the OFDM is employed as the signaling basis, then the optimal $\mathbf{V}$ may be expressed as
    \begin{equation}
        \mathbf{V}^\star = {\mathbf{U}^\star}^H\mathbf{F}_N^H = \operatorname{Diag}({\boldsymbol{{\theta}}})\mathbf{\Pi},
    \end{equation}
    in which case we have
    \begin{equation}
        \left\| {{{\mathbf{v}}_n} \odot {{\mathbf{v}}_m}} \right\|_2^2 = \delta_{n,m},
    \end{equation}
    owing to the structure of the permutation matrix $\mathbf{\Pi}$. As a consequence, we have $\left\|\mathbf{b}_k\right\|^2 = N$ as per \eqref{bk_def}. This suggests that $\left\|\mathbf{b}_k\right\|^2\le N$ is a supremum for every $k\ne 0$, and that
    \begin{equation}\label{sidelobe_inf}
        \mathbb{E}(|\tilde{r}_k|^2)\ge (\mu_4-1)N,\quad \forall k\ne 0,
    \end{equation}
    is an infimum of each individual sidelobe, both of which are achieved by the OFDM basis \eqref{ofdm_opt}. This indicates that OFDM achieves the lowest average sidelobe level at every lag.\footnote{In a different context, the DFT matrix has been shown as a maximizer of periodic autocorrelation sidelobe levels of the column vectors of a unitary matrix \cite{sarwate}. Here, we also show that OFDM is the unique maximizer up to complex permutations.}

    We may now prove the uniqueness of the OFDM as a global minimizer of every ranging sidelobe by contradiction. Suppose that there exists another signaling basis matrix $\mathbf{U}^{\prime}$ that achieves the infimum in \eqref{sidelobe_inf}, but has a different form from \eqref{ofdm_opt}. Then it holds immediately that $\mathbf{U}^{\prime}$ also minimizes the EISL. This contradicts Theorem \ref{thm:pacf_eisl}, which states that OFDM is the only EISL minimizer, completing the proof. 
\end{proof}

\textbf{Remark 3:} An interesting, but not surprising fact is that for PSK constellations ($\mu_4 = 1$), both the individual and integrated sidelobe levels are zero under OFDM signaling, which means that the P-ACF of OFDM-PSK signal is always a unit impulse function. This attribute has nothing to do with the randomness of the PSK, as it can be simply deduced from \eqref{P_ACF_expansion} by letting $\mathbf{V} = \mathbf{I}_N$ for any realization of PSK sequences, which originates from the duality of Fourier transform (FT), that the FT of a unit impulse function is a constant over all frequencies. Consequently, even for PSK constellations that do not satisfy Assumption 2, for example, BPSK, the sidelobe levels would also be zeros. This fact has been widely utilized in existing systems employing deterministic sensing sequences, for example, the DFT of PSK sequences have been utilized in the 5G Primary Synchronization Signal (PSS). However, as we shall see later, this does not hold for the A-ACF case, where PSK constellations always generate non-zero sidelobes.

\textbf{Remark 4:} While OFDM ensures zero ranging sidelobe levels for all PSK constellations and serves as the unique sidelobe minimizer under Assumption 2, it may not be unique in the case of BPSK. To illustrate this, we construct an alternative modulation basis, which is distinct from OFDM but also yields zero ranging sidelobes for BPSK. As discussed in the Fourier duality framework above, a necessary and sufficient condition for zero ranging sidelobes is that the ISAC signal must exhibit a constant amplitude spectrum. This requirement translates to:
\begin{equation}\label{freq_condition_BPSK}
\left|\mathbf{F}_N\mathbf{x}\right|^2=\left|\mathbf{F}_N\mathbf{U}\mathbf{s}\right|^2 = \left|\mathbf{V}^H\mathbf{s}\right|^2 = \mathbf{1}_N.
\end{equation}
In the BPSK case, this condition is clearly satisfied if $\mathbf{V}$ is a complex permutation matrix, which corresponds to the OFDM waveform. To construct an alternative modulation basis satisfying \eqref{freq_condition_BPSK}, we first consider the case of $N = 2$. Let
\begin{equation}
    \mathbf{V}_2^{\star} = \frac{1}{{\sqrt 2 }}\left[ {\begin{array}{*{20}{c}}
  1&j \\ 
  j&1 
\end{array}} \right].
\end{equation}
It is straightforward to verify that: $(i)$ $\mathbf{V}_2^{\star}$ is unitary, and $(ii)$ $\mathbf{V}_2^{\star H}\mathbf{s}$ produces two normalized QPSK symbols for any BPSK input vector $\mathbf{s}$, thereby satisfying the condition in \eqref{freq_condition_BPSK}.

This construction generalizes naturally for $N > 2$. When $N$ is even, we define:
\begin{equation}\label{V_construction_even}
    \mathbf{V}_N^{\star} = \mathbf{\Pi}\operatorname{Diag}({\boldsymbol{{\theta}}})\otimes\mathbf{V}_2,
\end{equation}
where $\mathbf{\Pi}\operatorname{Diag}({\boldsymbol{{\theta}}})$ is any complex permutation matrix of size $N/2$. For odd $N$, the corresponding construction is:
\begin{equation}\label{V_construction_odd}
    \mathbf{V}_N^{\star} = \left[ {\begin{array}{*{20}{c}}
  {{\mathbf{\Pi }}\operatorname{Diag} ({\mathbf{\theta }}) \otimes {{\mathbf{V}}_2}}&{{{\mathbf{0}}_{N - 1}}} \\ 
  {{\mathbf{0}}_{N - 1}^T}&1 
\end{array}} \right],
\end{equation}
where $\mathbf{\Pi}\operatorname{Diag}({\boldsymbol{{\theta}}})$ is any complex permutation matrix of size $(N-1)/2$. It can again be verified that both \eqref{V_construction_even} and \eqref{V_construction_odd} are unitary matrices satisfying the constant amplitude spectrum condition \eqref{freq_condition_BPSK} for the BPSK constellation. While additional constructions may exist, we do not pursue them here, as they fall outside the main scope of this work and are left for future investigation.

\begin{corollary}[Optimality of the CP-SC for Doppler Measurement]\label{coro:cpsc_Doppler}
   CP-SC is the only waveform that achieves the lowest Doppler sidelobe level for sub-Gaussian constellations in the P-ACF case.
\end{corollary}
\begin{proof}
We highlight that the Doppler sidelobe is generated from the ACF of the frequency spectrum of the signal, which is also known as the zero-delay slice of the ambiguity function. Also note that the Fourier transform switches the delay axis and the Doppler axis of the ambiguity function. Therefore, the result to be proved is a direct corollary of Theorem~\ref{thm:pacf_everywhere}, since CP-SC and OFDM is the Fourier-dual of each other, and OFDM is the unique waveform achieving the lowest delay sidelobe level for sub-Gaussian constellations in the P-ACF case.
\end{proof}

We now investigate the ranging performance of the $\mu_4 > 2$ case, i.e., when $\mathbf{s}$ is randomly realized from a super-Gaussian constellation.
\begin{corollary}\label{coro:cpsc_sg}
    CP-SC achieves the lowest EISL for super-Gaussian constellations in the P-ACF case.
\end{corollary}
\begin{proof}
    Let us recall \eqref{EISL_closed_form}. To minimize the EISL in the $\mu_4 > 2$ case, one needs to minimize $\left\|\mathbf{F}_N\mathbf{U}\right\|_4^4$ over the unitary group. Note that under a fixed $l_2$ norm, the $\ell_4$ norm is minimized if each entry has a constant modulus, which suggests that $\mathbf{U} = \mathbf{I}_N$ is a minimizer. This completes the proof.
\end{proof}

\textbf{Remark 5:} It is also interesting to point out that when $\mathbf{U} = \mathbf{I}_N$, we have
$\left\|\mathbf{F}_N\mathbf{U}\right\|_4^4 = 1$ and $\left\|\mathbf{b}_k\right\|^2 = 0$ for $k\ne 0$. This implies that using CP-SC signaling results in a ranging EISL of $N(N-1)$, and a uniform sidelobe level $\mathbb{E}(|\tilde{r}_k|^2) = N$ at each index $k \ne 0$, no matter what value that $\mu_4$ takes. Moreover, we see that the Gaussian constellation with $\mu_4 = 2$ produces exactly the same EISL and individual sidelobe level regardless of the choice of $\mathbf{U}$, which is consistent with its unitary invariance property. This observation suggests that under the CP-SC signaling, all the constellations behave like Gaussian in terms of ranging. In such a case, the only difference lies in their mainlobe levels, where the super-Gaussian constellation may slightly outperform other constellations, since a larger kurtosis yields higher mainlobe as per \eqref{mainlobe_level}. 

\section{The A-ACF Case}\label{sec_4}
In this section, we analyze the ranging performance of random ISAC signals by investigating the sidelobe levels of the A-ACF. 

\subsection{Main Results}
Let us first give a closed-form characterization of the average squared A-ACF.
\begin{corollary}\label{coro:aacf}
The average squared A-ACF is
    \begin{align}\label{average_squared_A-ACF}
        \mathbb{E}(|{r}_k|^2) = N^2\delta_{0,k}+N-k+(\mu_4-2)\sum\limits_{n = 1}^{N}|\mathbf{u}_n^H\mathbf{J}_k\mathbf{u}_n|^2.
    \end{align} 
\begin{proof}
Note that the A-ACF can be viewed as the P-ACF of zero-padded sequences, which implies that $\mathbb{E}(|{r}_k|^2)$ can be obtained as a specific case of $\mathbb{E}(|\tilde{r}_k|^2)$ given in Proposition \ref{prop:pacf} with $\mathbf{U}$ replaced by $[\mathbf{U}; \mathbf{0}_{N\times N}]$. With this substitution, \eqref{average_squared_A-ACF} can be obtained after some algebra. For a detailed derivation, please refer to Appendix \ref{prop_3_proof}.
\end{proof} 
\end{corollary}

Note that the average mainlobe level of the A-ACF stays the same with that of the P-ACF as in Corollary \ref{coro: mainlobe}. To proceed, we provide the following result on the EISL of the A-ACF.
\begin{corollary}\label{coro:aacf_eisl}
The EISL of the A-ACF is
    \begin{equation}\label{A-ACF_EISL_closed_form}
        \sum\limits_{k = 1}^{N - 1}\mathbb{E}({|r_k|^2}) = \frac{N(N-1)}{2} + (\mu_4 - 2)N(\left\|\tilde{\mathbf{F}}_{2N}\mathbf{U}\right\|_4^4-\frac{1}{2}),
    \end{equation}
    where $\tilde{\mathbf{F}}_{2N}\in \mathbb{C}^{2N \times N}$ contains the first $N$ columns of the size-$2N$ DFT matrix $\mathbf{F}_{2N}$.
\begin{proof}
This result may be viewed as a corollary of Proposition \ref{prop:eisl}, with the $N$-point DFT matrices replaced by the trimmed $2N$-point DFT matrix $\tilde{\mathbf{F}}_{2N}$. For a detailed derivation, please refer to Appendix \ref{prop_4_proof}.
\end{proof}
\end{corollary}

\subsection{Optimal Signaling Schemes}
In order to minimize the EISL for sub-Gaussian constellations ($\mu_4 < 2$), one has to solve the following $\ell_4$ norm maximization problem:
\begin{equation}\label{L4_max_hard}
    \max\limits_{\mathbf{U} \in \mathbb{U}(N)}\;\; \left\|\tilde{\mathbf{F}}_{2N}\mathbf{U}\right\|_4^4.
\end{equation}
Unfortunately, solving \eqref{L4_max_hard} is much more difficult than solving its counterpart \eqref{L4_max_easy}, since it is not possible to produce a complex permutation matrix through multiplying $\tilde{\mathbf{F}}_{2N}$ with a size-$N$ unitary matrix. While substituting the size-$N$ IDFT matrix indeed yields a large objective value, we do not know whether it is a globally optimal solution. In fact, maximizing the $\ell_4$ norm $\left\|{\mathbf{Z}}\mathbf{U}\right\|_4$ over the unitary group for arbitrary $\mathbf{Z}$ generally remains an open problem \cite{zhai2020complete,9518255}. Towards that end, we seek to establish the local optimality of the IDFT matrix for problem \eqref{L4_max_hard}, which in turn offers theoretical guarantee for the optimality of the OFDM signaling in the A-ACF case. 

\begin{thm}[OFDM is locally optimal for sub-Gaussian constellations]\label{thm:local_opt}
For the A-ACF case, OFDM constitutes a local minimum of the EISL for sub-Gaussian constellations, namely, $\mathbf{U}^\star = \mathbf{F}_N^H$ is a local maximizer of \eqref{L4_max_hard} over the unitary group $\mathbb{U}(N)$.
\end{thm}
\begin{proof}
    Observe that problem \eqref{L4_max_hard} can be reformulated as maximizing the function
    \begin{equation}
    f(\mathbf{V}) = \left\|\tilde{\mathbf{F}}_{2N}\mathbf{F}_N^{H}\mathbf{V}^H\right\|_4^4,
    \end{equation}
    subject to $\mathbf{V} \in \mathbb{U}(N)$. Therefore, the original problem reduces to proving that $\mathbf{V} = \mathbf{I}$ is a local maximizer of $f$. Given that the unitary group $\mathbb{U}(N)$ is a Riemannian manifold, proving local optimality at $\mathbf{I}$ entails demonstrating two conditions: $(i)$ The Riemannian gradient of $f$ is zero at $\mathbf{V} = \mathbf{I}$, and $(ii)$ $f$ is geodesically concave in a neighborhood of $\mathbf{I}$. For a detailed proof, please refer to Appendix \ref{thm_3_proof}.
\end{proof}

\textbf{Remark 6:} We note that a similar $\ell_4$ norm maximization problem was formulated in \cite{9518255} to construct a sparsity dictionary for beam-space mmWave signal processing. The problem was solved using a Coordinate Ascent (CA) algorithm, where unitary matrices were expressed as products of Givens rotation matrices and phase shift matrices. The authors showed that the DFT matrix is locally optimal with respect to the Givens rotation angle at each CA iteration. However, this finding does not fully guarantee the local optimality of the DFT matrix for the original optimization problem. In contrast, our proof directly demonstrates that the gradient and Hessian of \eqref{L4_max_hard} at $\mathbf{U}^\star = \mathbf{F}_N^H$ are zero and negative semi-definite, respectively, thus providing a more rigorous and complete proof compared to that of \cite{9518255}.

While the global optimality of the OFDM waveform under sub-Gaussian constellations is not theoretically guaranteed, our numerical results suggest that no other CP-free waveform achieves a lower sidelobe level than OFDM. Motivated by this empirical observation, we propose the following conjecture for future investigation.
\begin{conjecture}
    OFDM is the global minimizer of the EISL among all signals without CP under sub-Gaussian constellations. Specifically, $\mathbf{U}^\star = \mathbf{F}_N^H$ is the globally optimal solution to the $\ell_4$-norm maximization problem defined in \eqref{L4_max_hard}.
\end{conjecture}

Finally, we establish the global optimality of the SC waveform for super-Gaussian constellations in the A-ACF case.
\begin{corollary}\label{coro:aacf_sg}
    SC waveform achieves the lowest EISL for super-Gaussian constellations in the A-ACF case.
\end{corollary}
\begin{proof}
    Again, to minimize the A-ACF sidelobe for the $\mu_4>2$ case, one has to minimize {{$\|\tilde{\mathbf{F}}_{2N}\mathbf{U}\|_4^4$}} over the unitary group $\mathbb{U}(N)$. Since the minimum of the $\ell_4$ norm is attained if each entry of $\tilde{\mathbf{F}}_{2N}\mathbf{U}$ has a constant modulus, $\mathbf{U} = \mathbf{I}_N$ is a global minimizer, completing the proof.
\end{proof}

\section{Numerical Results}\label{sec_5}
In this section, we present numerical results to validate our theoretical analysis. In general, we compare the performance of three signaling schemes, namely, SC, OFDM, and CDMA waveforms employing Walsh codes. Note that we do not compare with OTFS since it is a 2-dimensional modulation over both delay and Doppler domains, while this paper considers 1-dimensional modulation and delay sidelobe analysis only. Due to the strict page limit, we designate the discussion of the Doppler sidelobe as our future work. We employ both QAM and PSK as examples for sub-Gaussian constellations. To construct a super-Gaussian constellation, we design a specific 64-APSK constellation consisting of 4 circles, each circle contains 16 points, whose radii are $4.54 \times 10^{-5}, 0.0067, 0.0815$ and $1.9983$, respectively, leading to a kurtosis of $3.9867$, and is referred to as ``SG-64-APSK'' in the legends of the figures. Moreover, all the simulation results are attained by averaging over $1000$ random realizations.

\begin{figure}[!t]
	\centering
	\includegraphics[width = \columnwidth]{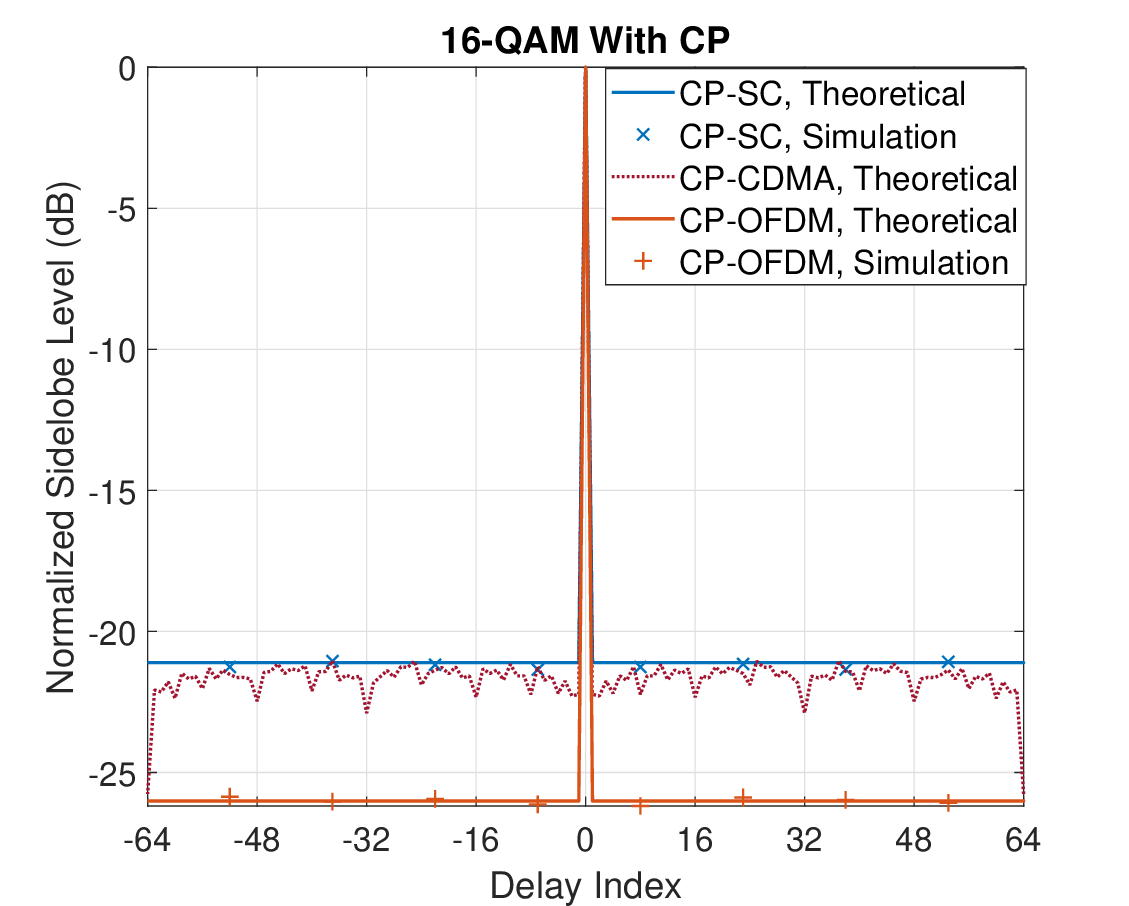}
	\caption{The P-ACF of the 16-QAM constellation under various signaling schemes, $N = 128$.}
    \label{16-QAM With CP}
\end{figure} 

\subsection{Average Sidelobe Level Analysis}
We first illustrate the P-ACF of the 16-QAM in Fig. \ref{16-QAM With CP} with $N = 128$, for signaling schemes with CP addition. It can be observed that the theoretical values matched very well with the simulation results, which confirms the correctness of our derivation in Sec. \ref{sec_3}. The CP-CDMA scheme with Hadamard matrix as signaling basis performs slightly better than the CP-SC approach. Moroever, CP-OFDM achieves 5 dB sidelobe level reduction compared to the CP-SC, and also outperforms the CP-CDMA.

\begin{figure}[!t]
	\centering
	\includegraphics[width = \columnwidth]{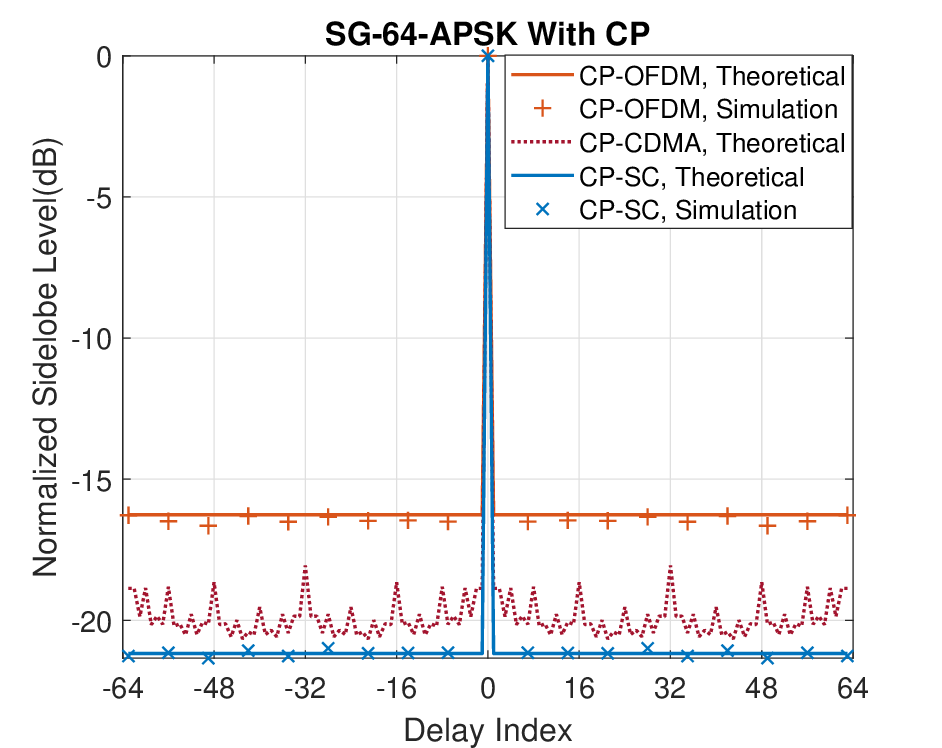}
	\caption{The P-ACF of the 64-APSK constellation under various signaling schemes, $N = 128$.}
    \label{64-APSK With CP}
\end{figure} 

To examine the performance of the super-Gaussian constellations, we portray the P-ACF of the designed SG-64-APSK constellation in Fig. \ref{64-APSK With CP} under various signaling schemes. As predicted by our theoretical framework, the CP-SC now becomes the best signaling scheme among all other strategies, which attains a 5 dB sidelobe reduction compared to the CP-OFDM waveform. In fact, even the CP-CDMA signal yields a lower sidelobe level than the CP-OFDM counterpart, making the latter the worst signaling basis for the super-Gaussian constellations.

\begin{figure}[!t]
	\centering
	\includegraphics[width = \columnwidth]{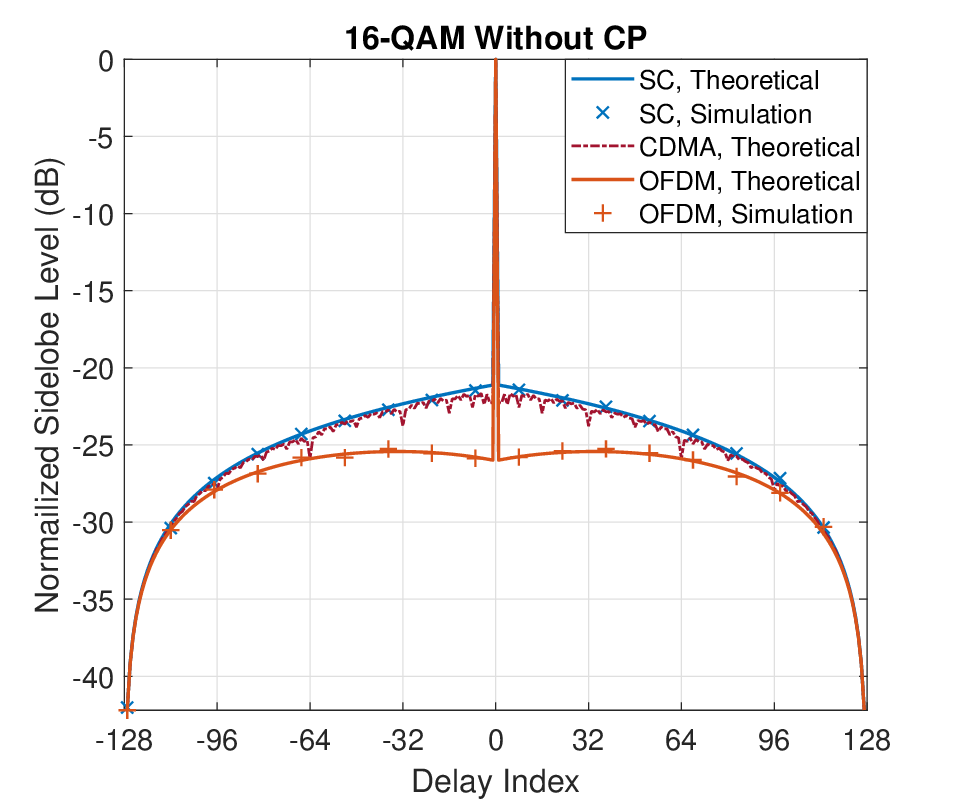}
	\caption{The A-ACF of the 16-QAM constellation under various signaling schemes, $N = 128$.}
    \label{16-QAM Without CP}
\end{figure} 

\begin{figure}[!t]
	\centering
	\includegraphics[width = \columnwidth]{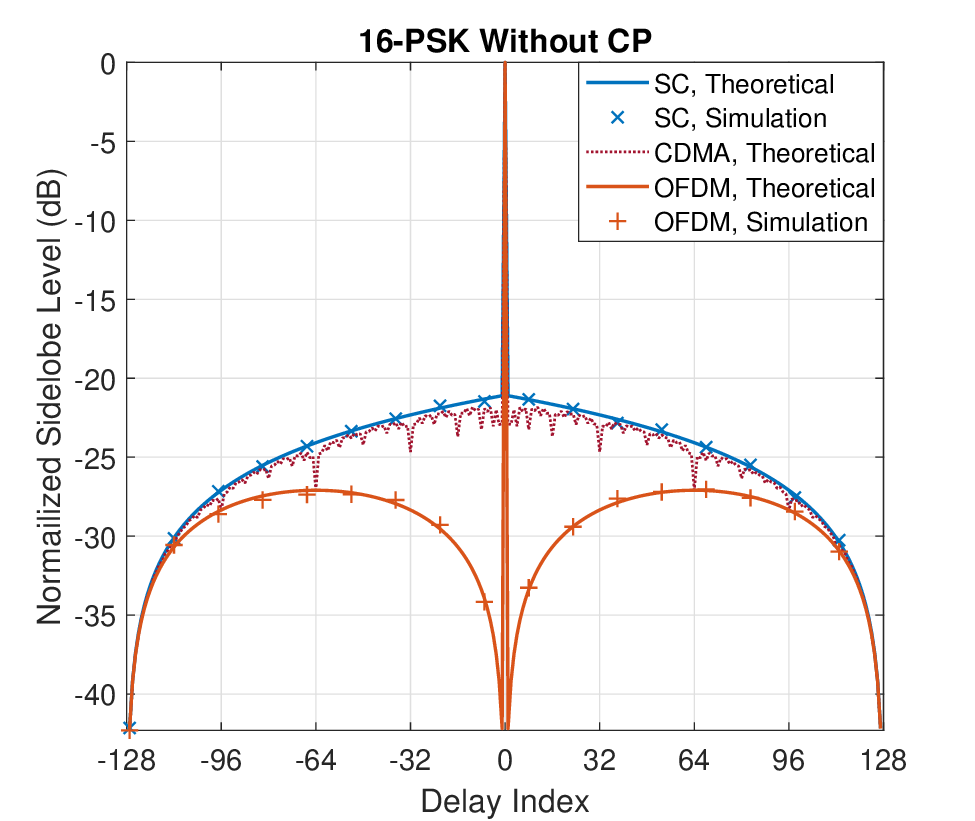}
	\caption{The A-ACF of the 16-PSK constellation under various signaling schemes, $N = 128$.}
    \label{16-PSK Without CP}
\end{figure} 

We then look at the sidelobe performance of 16-QAM and 16-PSK constellations for signaling strategies without CP in Fig. \ref{16-QAM Without CP} and Fig. \ref{16-PSK Without CP}. Again, all the theoretical results perfectly match their numerical counterparts. Moreover, the same trends may be observed in both figures, that the OFDM is superior to both SC and CDMA schemes, even if there is no global optimality guarantee for OFDM in the CP-free case. It is also interesting to highlight that for the 16-PSK constellation with OFDM signaling, the sidelobe level goes down when the delay index approaches to zero. This is because when the delay is small, the linear convolution may be approximated as a periodic convolution. In such a case, the sidelobe value of the A-ACF at those delay lags may be very close to those of its P-ACF counterpart, which is exactly zero for PSK alphabets, as discussed in Remark 3. As a comparison, the sidelobe performance of the SG-64-APSK constellation is portrayed in Fig. \ref{64-APSK Without CP}, under three signaling strategies without CP. It can be clearly seen that the SC scheme reaches the lowest sidelobe compared to other two waveforms, which is consistent with Corollary \ref{coro:cpsc_sg}. 

\begin{figure}[!t]
	\centering
	\includegraphics[width = \columnwidth]{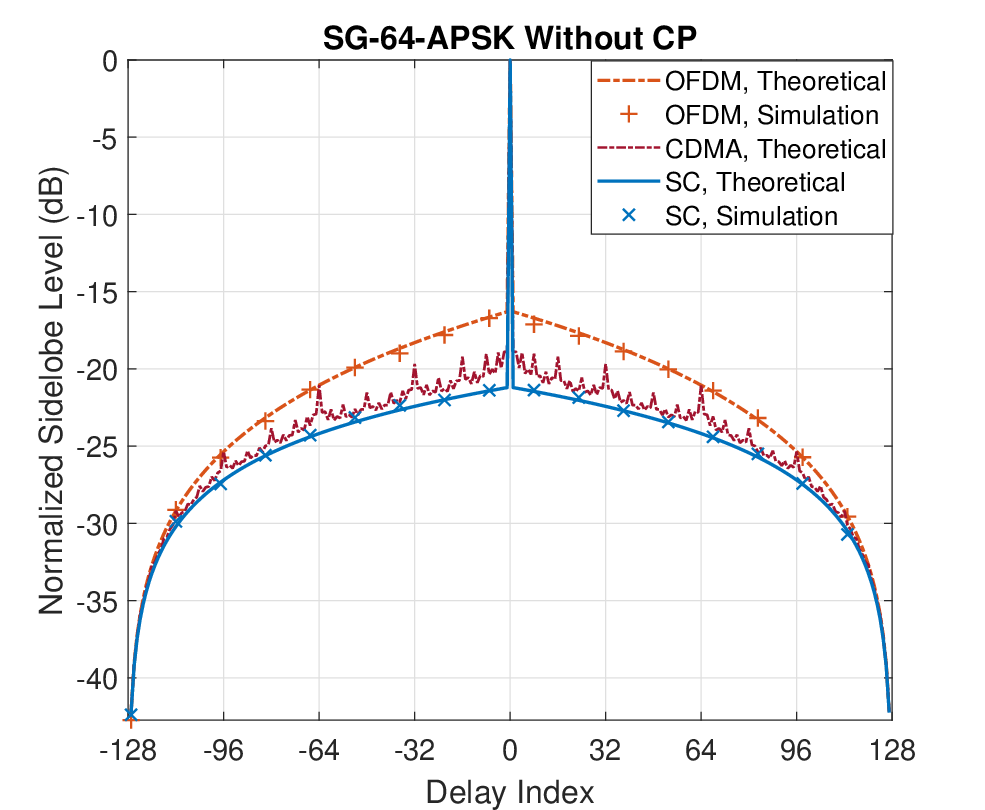}
	\caption{The A-ACF of the 64-APSK constellation under various signaling schemes, $N = 128$.}
    \label{64-APSK Without CP}
\end{figure}

\begin{figure}[!t]
	\centering
	\includegraphics[width = \columnwidth]{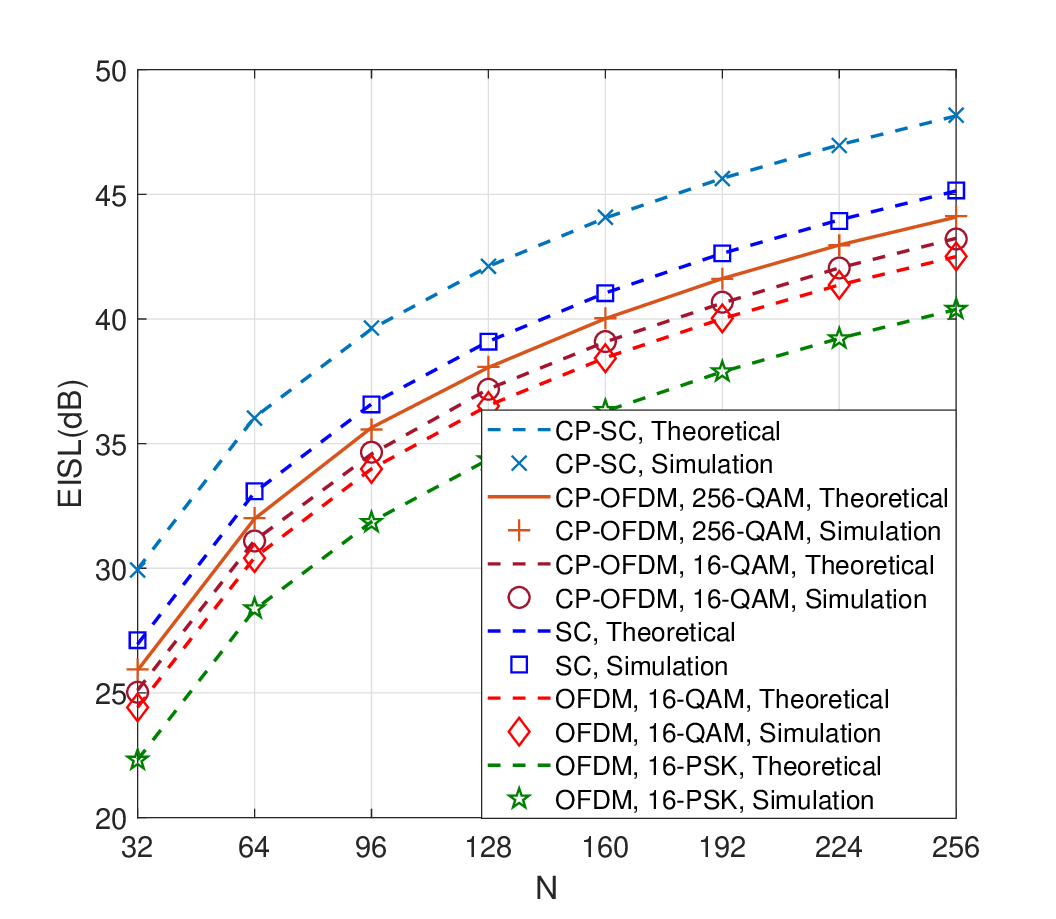}
	\caption{The resultant EISL for different constellations under OFDM and SC signaling with varying number of symbols.}
    \label{EISL_Comparison}
\end{figure} 

To show the overall performance of different signaling schemes, we illustrate the resultant EISL for different constellations under SC and OFDM waveforms in Fig. \ref{EISL_Comparison}, where both cases with and without CP are considered. First of all, signals with CP generally lead to higher sidelobe levels compared to those without, which may also be inferred from their respective EISL expressions \eqref{EISL_closed_form} and \eqref{A-ACF_EISL_closed_form}. Moreover, in both P-ACF and A-ACF cases, SC schemes result in the same EISL regardless of the choice of constellations. Higher-order QAM modulations always end up with larger sidelobe levels, owing to a larger kurtosis. As expected, PSK with OFDM signaling attains the smallest EISL of the A-ACF, despite that OFDM is only a local optimum in such a case.

\subsection{Variance of the Integrated Sidelobe and Mainlobe Level}
Although EISL characterizes the average behavior of the ACF, it does not reflect variability in the sidelobe level, which may critically affect detection reliability. To address this limitation, we introduce a complementary metric, namely, the Variance of Integrated Sidelobe and Mainlobe Level (VISML), for both the P-ACF and A-ACF, defined as:
\begin{align}
    \nonumber&\text{VISML}_\text{P-ACF} = \operatorname{Var}\left(\sum\limits_{k = 0}^{N - 1}{|\tilde{r}_k|^2}\right),\\
    &\text{VISML}_\text{A-ACF}= \operatorname{Var}\left(\sum\limits_{k = -N+1}^{N - 1}{|{r}_k|^2}\right).
\end{align}
This metric captures the fluctuations in both mainlobe and sidelobe power, which are crucial for ensuring stable and robust sensing performance.

By applying Parseval’s theorem and straightforward calculations, we obtain:
\begin{align}
     \nonumber \sum\limits_{k=0}^{N-1}\left|\tilde{r}_k\right|^2 &= N\left\|\mathbf{F}_N\mathbf{U}\mathbf{s}\right\|_4^4,\\
    \sum\limits_{k = -N+1}^{N - 1}{|{r}_k|^2} &= 2N\left\|\tilde{\mathbf{F}}_{2N}\mathbf{U}\mathbf{s}\right\|_4^4.
\end{align}
Accordingly, the VISML expressions for the two types of ACF become:
\begin{align}\label{VISML}
\nonumber&\text{VISML}_\text{P-ACF} = N^2\left[\mathbb{E}\left(\left\|\mathbf{F}_N\mathbf{U}\mathbf{s}\right\|_4^8\right) - \mathbb{E}^2\left(\left\|\mathbf{F}_N\mathbf{U}\mathbf{s}\right\|_4^4\right)\right],\\
&\text{VISML}_\text{A-ACF} = 4N^2\left[\mathbb{E}\left(\left\|\tilde{\mathbf{F}}_{2N}\mathbf{U}\mathbf{s}\right\|_4^8\right) -  \mathbb{E}^2\left(\left\|\tilde{\mathbf{F}}_{2N}\mathbf{U}\mathbf{s}\right\|_4^4\right)\right].
\end{align}
However, computing the VISML in closed form is generally intractable due to the involvement of up to the 8th-order moments of $\mathbf{s}$ and numerous cross terms. Even if an analytical expression were derived, proving that OFDM globally minimizes the VISML still remains highly challenging due to the non-convex dependence on $\mathbf{U}$ and the entangled structure of high-order statistics. To overcome these difficulties, we develop a numerical optimization approach to compute the signaling basis that minimizes VISML. Specifically, we adopt a stochastic gradient projection (SGP) method, which iteratively updates the solution in the direction of the negative stochastic gradient and projects the iterating point onto the unitary group to enforce orthonormal constraints.

\begin{figure}[!t]
	\centering
	\includegraphics[width = \columnwidth]{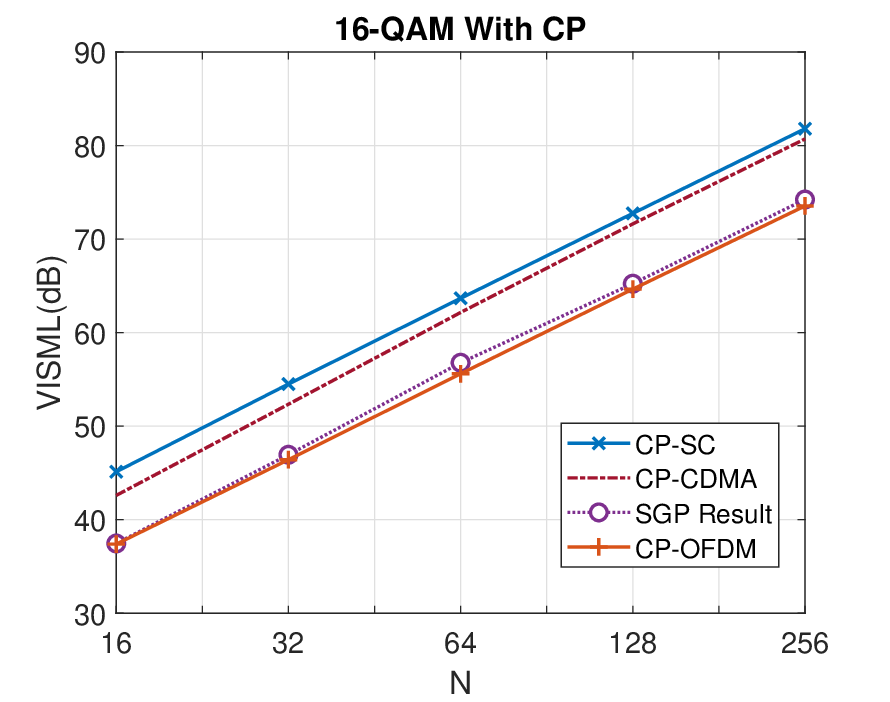}
	\caption{The VISML of the 16-QAM constellation under SC, OFDM, and SGP-based modulation schemes with CP.}
    \label{16-QAM With CP VISML SGP}
\end{figure} 

\begin{figure}[!t]
	\centering
	\includegraphics[width = \columnwidth]{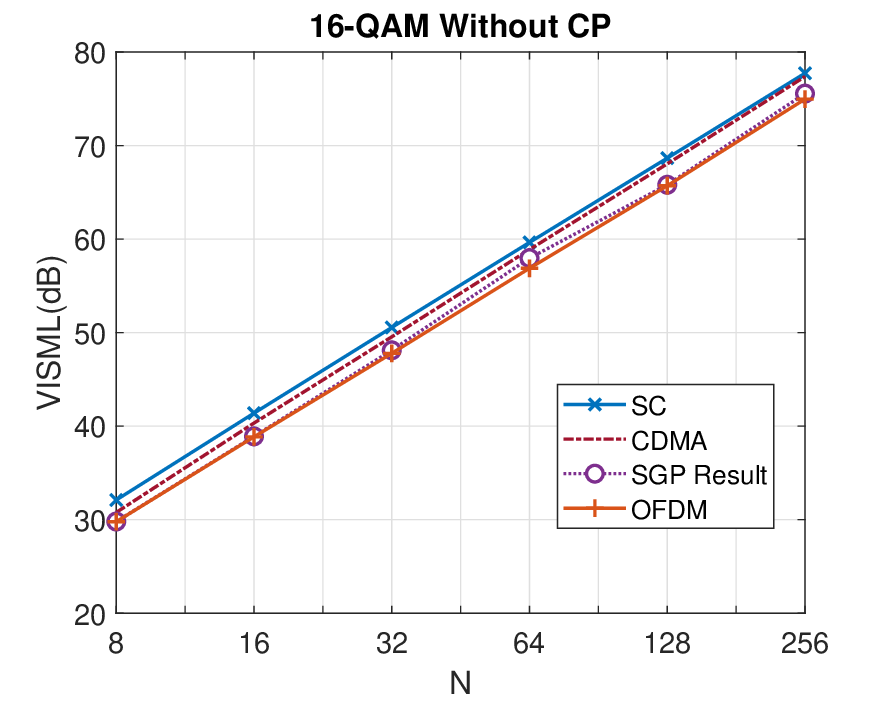}
	\caption{The VISML of the 16-QAM constellation under SC, OFDM, and SGP-based modulation schemes without CP.}
    \label{16-QAM Without CP VISML SGP}
\end{figure} 

In light of the above discussion, let us evaluate the VISML performance of different modulation schemes in terms of both P-ACF and A-ACF, which are depicted in Figs. \ref{16-QAM With CP VISML SGP} and \ref{16-QAM Without CP VISML SGP}, respectively, for increasing $N$ under a 16-QAM constellation. Notably, OFDM achieves the lowest overall variance among all waveforms, both with and without CP. To further validate OFDM’s optimality, we also present VISML results obtained via the proposed SGP algorithm for 16-QAM. In both P-ACF and A-ACF scenarios, the VISML values resulting from the SGP method closely match those of OFDM. Remarkably, the convergence point of the SGP algorithm consistently leads to (nearly) a complex permutation matrix of $\mathbf{V} = \mathbf{U}^H\mathbf{F}_N^H$, aligning with the structure of the $\ell_4$-norm maximization solution\footnote{Due to finite numerical precision, the converged matrix $\mathbf{V}$ does not exhibit exact zeros in non-permutation entries, resulting in a small performance gap compared to the OFDM.}. These findings indicate that OFDM corresponds to at least a local minimum of the VISML objective, reinforcing its significance as an optimal communication-centric ISAC waveform for ranging. Nevertheless, a rigorous theoretical proof of OFDM’s global optimality in minimizing the VISML remains an open problem, which we pose as a conjecture for future work.

\subsection{Ranging Performance Analysis}
\begin{figure}[!t]
	\centering
	\includegraphics[width = \columnwidth]{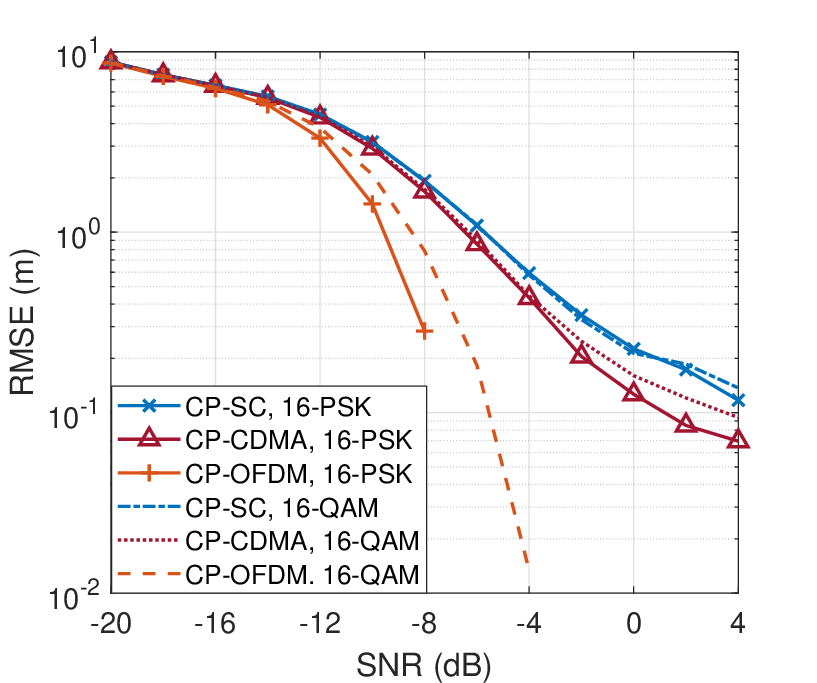}
	\caption{Two-target ranging RMSE of different waveforms with CP for PSK and QAM modulations, with a pair of strong and weak targets located at $11.25\text{m}$ and $18.75\text{m}$.}
    \label{RMSE_CP_results}
\end{figure} 

\begin{figure}[!t]
	\centering
	\includegraphics[width = \columnwidth]{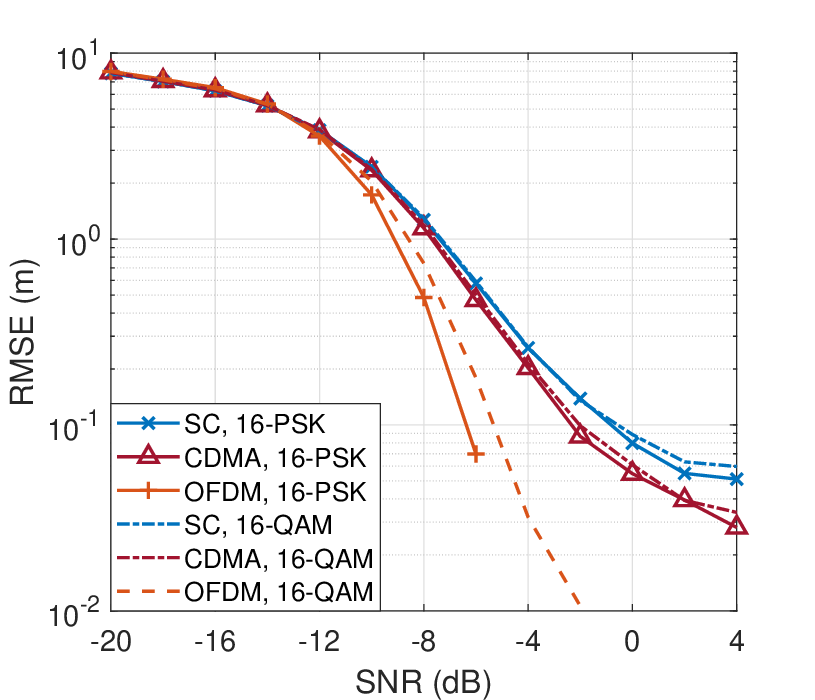}
	\caption{Two-target ranging RMSE of different waveforms without CP for PSK and QAM modulations, with a pair of strong and weak targets located at $11.25\text{m}$ and $18.75\text{m}$.}
    \label{RMSE_CPfree_results}
\end{figure} 
Finally, we examine the practical ranging performance of different waveforms under random PSK and QAM symbols in Fig. \ref{RMSE_CP_results} and Fig. \ref{RMSE_CPfree_results}. In particular, we consider a ranging task where a strong target and a weak target need to be simultaneously sensed. The bandwidth of different waveforms is set as $800\;\text{MHz}$. By fixing the transmit power to 1, the SNR is defined as the inverse of the noise variance. The two targets are located at $11.25\text{m}$ and $18.75\text{m}$, respectively, where the reflection power of the $11.25\text{m}$ target is 10 dB higher than the one at $18.75\text{m}$, such that the latter may be masked by the sidelobe of the strong target with high probability. In this sense, a lower EISL intuitively indicates a better ranging performance. This intuition has been confirmed by our simulation, where the OFDM waveform always outperforms the SC and CDMA by an order of magnitude, for both cases with and without CP. It may also be observed that PSK achieves significantly lower ranging errors compared to QAM for CP-OFDM/OFDM signaling, and a slightly better performance in the CP-CDMA/CDMA cases, despite that they exhibit almost equivalent performance under SC waveforms. These results are all consistent with their respective sidelobe levels analyzed in the above, which guarantee the usefulness of the proposed EISL as a well-defined ranging metric for random ISAC signaling.

\section{Conclusions}\label{sec_6}
This study has provided a comprehensive analysis of communication-centric ISAC waveforms, specifically focusing on their sensing performance in terms of the ranging sidelobe levels. Our findings demonstrated the superiority of OFDM modulation over other waveforms in achieving the lowest ranging sidelobe, confirmed through rigorous evaluation of both aperiodic and periodic auto-correlation functions. The introduction of the expectation of the integrated sidelobe level (EISL) as a key metric has further quantified this performance, establishing OFDM as the globally optimal waveform in the presence of a cyclic prefix (CP) and a locally optimal waveform in the absence of CP. The theoretical proofs and numerical validations presented reinforce OFDM's pivotal role in enhancing the ranging performance of ISAC systems. Future work should explore potential enhancements in waveform design, for example, nonlinear modulation techniques providing stronger sidelobe level guarantees, and further refine these findings by considering practical power allocation and pulse shaping designs. 

\section*{Acknowledgment}
The authors would like to thank the anonymous Reviewers and the Editor for their insightful comments and constructive suggestions, which have greatly improved the quality and clarity of this paper.

\appendices


\section{Proof of Proposition \ref{prop:pacf}}\label{prop_1_proof}
By noting \eqref{P_ACF_expansion}, the squared P-ACF may be formulated as
    \begin{align}
        |\tilde{r}_k|^2 \nonumber&= \sum\limits_{n = 1}^{N}|\mathbf{v}_n^H\mathbf{s}|^2 e^{\frac{-j2\pi k(n-1)}{N}}\sum\limits_{m = 1}^{N}|\mathbf{v}_m^H\mathbf{s}|^2 e^{\frac{j2\pi k(m-1)}{N}}\\
        &=\sum\limits_{n = 1}^{N}\sum\limits_{m = 1}^{N}|\mathbf{v}_n^H\mathbf{s}|^2|\mathbf{v}_m^H\mathbf{s}|^2e^{\frac{-j2\pi k(n-m)}{N}}.
    \end{align}
Expanding $|\mathbf{v}_n^H\mathbf{s}|^2$ yields
\begin{align}
    |\mathbf{v}_n^H\mathbf{s}|^2 \nonumber&= \mathbf{v}_n^H\mathbf{s}\mathbf{s}^H\mathbf{v}_n= (\mathbf{v}_n^T \otimes \mathbf{v}_n^H)\operatorname{vec}(\mathbf{s}\mathbf{s}^H)\\
    &=(\mathbf{v}_n^T \otimes \mathbf{v}_n^H)\tilde{\mathbf{s}} = \tilde{\mathbf{s}}^H(\mathbf{v}_n^* \otimes \mathbf{v}_n),
\end{align}
where $\tilde{\mathbf{s}}\triangleq\operatorname{vec}(\mathbf{s}\mathbf{s}^H)$. Therefore,
\begin{align}\label{average_P-ACF_step_1}
    \mathbb{E}(|\tilde{r}_k|^2) = \sum\limits_{n = 1}^{N}\sum\limits_{m = 1}^{N}(\mathbf{v}_n^T \otimes \mathbf{v}_n^H)\mathbf{S}(\mathbf{v}_m^* \otimes \mathbf{v}_m)e^{\frac{-j2\pi k(n-m)}{N}},
\end{align}
where $\mathbf{S} = \mathbb{E}(\tilde{\mathbf{s}}\tilde{\mathbf{s}}^H)$, whose entries are given by
\begin{equation}\label{s_tensor}
\begin{gathered}
  S_{(m-1)N+p,(n-1)N+q} = \mathbb{E}(s_m^*{s_n}{s_p}s_q^*) =  \hfill \\
  \left\{ {\begin{array}{*{20}{l}}
  {\mathbb{E}({{\left| {{s_n}} \right|}^4}) = {\mu _4},\;m = n = p = q}, \\ 
  {\mathbb{E}({{\left| {{s_n}} \right|}^2})\mathbb{E}({{\left| {{s_p}} \right|}^2}) = 1,\;m = n,\;p = q,\;n \ne p}, \\ 
  {\mathbb{E}({{\left| {{s_m}} \right|}^2})\mathbb{E}({{\left| {{s_n}} \right|}^2}) = 1,\;m = p,\;n = q,\;m \ne n}, \\ 
  {\mathbb{E}(s_m^{*2})\mathbb{E}(s_n^2) = 0,\;m = q,\;n = p,\;m \ne n\;}, \\ 
  {0,\;{\text{otherwise}}}, 
\end{array}} \right. \hfill \\ 
\end{gathered} 
\end{equation}
yielding
\begin{align}
\mathbf{S} \nonumber&= \mathbb{E}(\tilde{\mathbf{s}}\tilde{\mathbf{s}}^H) \\
  &= \left[ {\begin{array}{*{20}{c}}
  {{\mu _4}}&{{\mathbf{0}}_N^T}&1&{{\mathbf{0}}_N^T}&1& \cdots &1 \\ 
  {{{\mathbf{0}}_N}}&{{{\mathbf{I}}_N}}&{{{\mathbf{0}}_N}}&{{{\mathbf{0}}_N}}&{{{\mathbf{0}}_N}}& \cdots &{{{\mathbf{0}}_N}} \\ 
  1&0&{{\mu _4}}& \ldots &1& \ldots &1 \\ 
  {{{\mathbf{0}}_N}}&{{{\mathbf{0}}_N}}&{{{\mathbf{0}}_N}}&{{{\mathbf{I}}_N}}& \vdots & \cdots & \vdots  \\ 
  1&0&1&0&{{\mu _4}}& \cdots &1 \\ 
   \vdots & \vdots & \vdots & \vdots & \vdots & \ddots &{{{\mathbf{0}}_N}} \\ 
  1&0&1&0&1& \cdots &{{\mu _4}} 
\end{array}} \right] \in \mathbb{R}^{N^2\times N^2},
\end{align}
with $\mathbf{0}_{N}$ representing the all-zero vector with length $N$.

To simplify \eqref{average_P-ACF_step_1}, $\mathbf{S}$ is decomposed as
\begin{equation}\label{S_decompose}
    \mathbf{S} = \mathbf{I}_{N^2} + \mathbf{S}_1 + \mathbf{S}_2,
\end{equation}
where 
\begin{align}
    {{\mathbf{S}}_1} &= \operatorname{Diag} \left( {{{\left[ {{\mu _4}-2,{\mathbf{0}}_N^T,{\mu _4}-2,{\mathbf{0}}_N^T, \ldots {\mu _4}-2} \right]}^T}} \right),\\
    {{\mathbf{S}}_2} &= \left[ {{\mathbf{c}},{{\mathbf{0}}_{{N^2} \times N}},{\mathbf{c}}, \ldots ,{\mathbf{c}},{{\mathbf{0}}_{{N^2} \times N}},{\mathbf{c}}} \right],
\end{align}
with ${{\mathbf{0}}_{{N^2} \times N}}$ being the all-zero matrix of size ${{N^2} \times N}$, and 
\begin{equation}
    {\mathbf{c}} = {\left[ {1,{\mathbf{0}}_N^T,1, \ldots ,1,{\mathbf{0}}_N^T,1} \right]^T}.
\end{equation}
Due to the fact that $\mathbf{v}_n^H\mathbf{v}_m = \delta_{n,m}$, we have
\begin{align}
    &(\mathbf{v}_n^T \otimes \mathbf{v}_n^H)\mathbf{I}_{N^2}(\mathbf{v}_m^* \otimes \mathbf{v}_m) = \mathbf{v}_n^T\mathbf{v}_m^*\mathbf{v}_n^H\mathbf{v}_m = \delta_{n,m}, \label{1st_term}\\\nonumber
    &(\mathbf{v}_n^T \otimes \mathbf{v}_n^H)\mathbf{S}_1(\mathbf{v}_m^* \otimes \mathbf{v}_m)=  (\mu_4 -2)\sum\limits_{p = 1}^{N}|v_{p,n}|^2|v_{p,m}|^2 \\
    &=(\mu_4 -2)\left\| {{{\mathbf{v}}_n} \odot {{\mathbf{v}}_m}} \right\|_2^2, \label{2nd_term}\\\nonumber
    &(\mathbf{v}_n^T \otimes \mathbf{v}_n^H)\mathbf{S}_2(\mathbf{v}_m^* \otimes \mathbf{v}_m)=\sum\limits_{p = 1}^{N}|v_{p,n}|^2\sum\limits_{q = 1}^{N}|v_{q,m}|^2 \\
    & = \left\| {{\mathbf{v}}_n} \right\|_2^2\left\| {{\mathbf{v}}_m} \right\|_2^2 = 1,\label{3rd_term}
\end{align}
Plugging \eqref{S_decompose}, \eqref{1st_term}-\eqref{3rd_term} into \eqref{average_P-ACF_step_1} immediately leads to
\begin{align}\label{squared_P-ACF_1}
    \mathbb{E}(|\tilde{r}_k|^2) \nonumber&= N^2\delta_{0,k}+N\\
    &+(\mu_4-2)\sum\limits_{n = 1}^{N}\sum\limits_{m = 1}^{N}\left\| {{{\mathbf{v}}_n} \odot {{\mathbf{v}}_m}} \right\|_2^2 e^{\frac{j2\pi k(n-m)}{N}}.
\end{align}
Moreover, based on the definition of \eqref{bk_defined}, we have
\begin{align}\label{bk_def}
    \left\|\mathbf{b}_k\right\|_2^2 \nonumber&= \sum\limits_{p = 1}^{N}\sum\limits_{n = 1}^{N}\sum\limits_{m = 1}^{N}|v_{p,n}|^2|v_{p,m}|^2e^{\frac{j2\pi k(n-m)}{N}}\\
    & = \sum\limits_{n = 1}^{N}\sum\limits_{m = 1}^{N}\left\| {{{\mathbf{v}}_n} \odot {{\mathbf{v}}_m}} \right\|_2^2 e^{\frac{j2\pi k(n-m)}{N}}.
\end{align}
Therefore, \eqref{squared_P-ACF_1} can be recast in a compact form as \eqref{squared_P-ACF}, completing the proof.

\section{Proof of Proposition \ref{prop:eisl}}\label{prop_2_proof}
It can be noted from \eqref{bk_defined} that the $p$th entry of $\mathbf{b}_k$, namely,
    \begin{equation}
        b_{k,p} = \sum\limits_{n = 1}^N {{{\left| {{v_{p,n}}} \right|}^2}{e^{\frac{{ - j2\pi k\left( {n - 1} \right)}}{N}}}} 
    \end{equation}
    is the DFT of the $p$th row of $\mathbf{V}$. Using the Parseval's theorem yields
    \begin{equation}
        \frac{1}{N}\sum\limits_{k = 0}^{N-1}|b_{k,p}|^2 = \sum\limits_{n = 1}^N{{\left| {{v_{p,n}}} \right|}^4},
    \end{equation}
    and hence
    \begin{align}
         \sum\limits_{k = 0}^{N - 1}\mathbb{E}(|\tilde{r}_k|^2) \nonumber&= 2N^2+(\mu_4-2)\sum\limits_{k = 0}^{N - 1}\left\|\mathbf{b}_k\right\|_2^2\\
         \nonumber&=2N^2+(\mu_4-2)\sum\limits_{k = 0}^{N - 1}\sum\limits_{p = 1}^{N}|b_{k,p}|^2\\
         &= 2N^2+(\mu_4-2)N\left\|\mathbf{V}\right\|_4^4,
    \end{align}
    which implies
    \begin{align}
        \nonumber&\sum\limits_{k = 1}^{N - 1}\mathbb{E}({|\tilde{r}_k|^2}) =\sum\limits_{k = 0}^{N - 1}\mathbb{E}(|\tilde{r}_k|^2)-\mathbb{E}(|\tilde{r}_0|^2)\\
        &= N(N-1) + (\mu_4 - 2)N(\left\|\mathbf{F}_N\mathbf{U}\right\|_4^4-1).
    \end{align}

\section{Proof of Corollary \ref{coro:aacf}}\label{prop_3_proof}
For notational convenience, let us define $\mathbf{A} = \mathbf{U}^H = \left[\mathbf{a}_1, \mathbf{a}_2,\ldots,\mathbf{a}_N\right]$, such that the $n$th row of $\mathbf{U}$ may be denoted as $\mathbf{a}_n^H$, and thereby $a_{m,n} = u_{n,m}^*$. Therefore we have
\begin{equation}
     \sum\limits_{n = 1}^{N-k}a_{m,n}a^*_{m,n+k} = \sum\limits_{n = 1}^{N-k}u_{n,m}^*u_{n+k,m} = \mathbf{u}_m^H\mathbf{J}_k\mathbf{u}_m,
\end{equation}
which is the A-ACF of $\mathbf{u}_m$, and thus
\begin{align}
    \nonumber &\sum\limits_{n = 1}^{N-k}(\mathbf{a}^*_{n+k}\odot \mathbf{a}_n)\\
    &\nonumber = \left[\sum\limits_{n = 1}^{N-k}a_{1,n}a^*_{1,n+k},\ldots,\sum\limits_{n = 1}^{N-k}a_{N,n}a^*_{N,n+k}\right]^T\\
    & = \left[\mathbf{u}_1^H\mathbf{J}_k\mathbf{u}_1,\ldots,\mathbf{u}_N^H\mathbf{J}_k\mathbf{u}_N\right]^T,
\end{align}
which implies
\begin{equation}\label{A-ACF_U}
    \left\|\sum\limits_{n = 1}^{N-k}(\mathbf{a}^*_{n+k}\odot \mathbf{a}_n)\right\|_2^2 = \sum\limits_{n = 1}^{N}|\mathbf{u}_n^H\mathbf{J}_k\mathbf{u}_n|^2.
\end{equation}
Moreover, we have $\mathbf{x} = \mathbf{U}\mathbf{s} = \left[\mathbf{a}_1^H\mathbf{s}, \mathbf{a}_2^H\mathbf{s},\ldots,\mathbf{a}_N^H\mathbf{s}\right]^T$. The A-ACF of $\mathbf{x}$ may be therefore expressed as
\begin{equation}\label{A-ACF_expansion}
    r_k = \mathbf{x}^H\mathbf{J}_k\mathbf{x} = \mathbf{s}^H\mathbf{A}\mathbf{J}_k\mathbf{A}^H\mathbf{s} = \sum\limits_{n = 1}^{N-k}\mathbf{a}_n^H\mathbf{s}\mathbf{s}^H\mathbf{a}_{n+k}.
\end{equation}
With the above identities, we are now ready to present the average squared A-ACF in closed form.

From \eqref{A-ACF_expansion} we have
    \begin{align}\label{squared_A-ACF}
        |r_k|^2 \nonumber&= \sum\limits_{n = 1}^{N-k}\sum\limits_{m = 1}^{N-k}(\mathbf{a}_n^H\mathbf{s}\mathbf{s}^H\mathbf{a}_{n+k})(\mathbf{a}_{m+k}^H\mathbf{s}\mathbf{s}^H\mathbf{a}_m)\\
        &=\sum\limits_{n = 1}^{N-k}\sum\limits_{m = 1}^{N-k}(\mathbf{a}_{n+k}^T\otimes\mathbf{a}_n^H){\tilde{\mathbf{s}}}{\tilde{\mathbf{s}}}^T(\mathbf{a}_{m}\otimes\mathbf{a}_{m+k}^*),
    \end{align}
    where we define ${\tilde{\mathbf{s}}} = \operatorname{vec}(\mathbf{s}\mathbf{s}^H)$.
    In order to reuse the results in Proposition \ref{prop:pacf}, one has to transform ${\tilde{\mathbf{s}}}{\tilde{\mathbf{s}}}^T$ into ${\tilde{\mathbf{s}}}{\tilde{\mathbf{s}}}^H$. This can be realized by the commutation matrix $\mathbf{K} \in \mathbb{O}(N^2)$, such that $\mathbf{K}(\mathbf{a}\otimes\mathbf{b}) = \mathbf{b}\otimes\mathbf{a}$, where $\mathbb{O}(n)$ represents the orthogonal group of degree $n$. Therefore,
    \begin{equation}
       \mathbf{K}{\tilde{\mathbf{s}}} = \mathbf{K}(\mathbf{s}^*\otimes\mathbf{s}) = \mathbf{s}\otimes\mathbf{s}^* = {\tilde{\mathbf{s}}}^* \Rightarrow {\tilde{\mathbf{s}}}^T\mathbf{K}^T = {\tilde{\mathbf{s}}}^H.
    \end{equation}
    One may therefore recast each term in \eqref{squared_A-ACF} as
    \begin{equation}
    \begin{gathered}
      ({\mathbf{a}}_{n + k}^T \otimes {\mathbf{a}}_n^H)\tilde {\mathbf{s}}{\tilde {\mathbf{s}}^T}({{\mathbf{a}}_m} \otimes {\mathbf{a}}_{m + k}^*) \hfill \\
       = ({\mathbf{a}}_{n + k}^T \otimes {\mathbf{a}}_n^H)\tilde {\mathbf{s}}{\tilde {\mathbf{s}}^T}{{\mathbf{K}}^T}{\mathbf{K}}({{\mathbf{a}}_m} \otimes {\mathbf{a}}_{m + k}^*) \hfill \\
       = ({\mathbf{a}}_{n + k}^T \otimes {\mathbf{a}}_n^H)\tilde {\mathbf{s}}{\tilde {\mathbf{s}}^H}({\mathbf{a}}_{m + k}^* \otimes {{\mathbf{a}}_m}). \hfill \\ 
    \end{gathered} 
    \end{equation}
    Hence, the average squared A-ACF may be formulated as
    \begin{align}\label{a_squared_A-ACF}
        \mathbb{E}(|r_k|^2) = \sum\limits_{n = 1}^{N-k}\sum\limits_{m = 1}^{N-k}({\mathbf{a}}_{n + k}^T \otimes {\mathbf{a}}_n^H)\mathbf{S}({\mathbf{a}}_{m + k}^* \otimes {{\mathbf{a}}_m}),
    \end{align}
    with $\mathbf{S} = \mathbb{E}({\tilde {\mathbf{s}}}{\tilde {\mathbf{s}}^H})$. Using again \eqref{s_tensor}, we arrive at
    \begin{align}
     &(\mathbf{a}_{n+k}^T\otimes \mathbf{a}_n^H)\mathbf{S}(\mathbf{a}_{m+k}^*\otimes \mathbf{a}_m) \nonumber \\
     &\hspace{3mm}=\delta_{n,m} + \delta_{0,k}+(\mu_4-2)(\mathbf{a}_{n+k}^*\odot\mathbf{a}_n)^H(\mathbf{a}_{m+k}^*\odot \mathbf{a}_m),
    \end{align}
    which amounts to \eqref{average_squared_A-ACF}.

\section{Proof of Corollary \ref{coro:aacf_eisl}}\label{prop_4_proof}
It can be straightforwardly deduced from \eqref{average_squared_A-ACF} that
    \begin{align}\label{direct_sum}
        &\sum\limits_{k = 1}^{N - 1}\mathbb{E}({|r_k|^2}) \nonumber= \sum\limits_{k = 1}^{N - 1} (N-k) + (\mu_4 - 2)\sum\limits_{k = 1}^{N - 1}\sum\limits_{n = 1}^{N}|\mathbf{u}_n^H\mathbf{J}_k\mathbf{u}_n|^2\\
        &= \frac{N(N-1)}{2} + (\mu_4 - 2)\sum\limits_{n = 1}^{N}\sum\limits_{k = 1}^{N-1}|\mathbf{u}_n^H\mathbf{J}_k\mathbf{u}_n|^2,
    \end{align}
    where $\sum\nolimits_{k = 1}^{N-1}|\mathbf{u}_n^H\mathbf{J}_k\mathbf{u}_n|^2$ is the ISL of the A-ACF of $\mathbf{u}_n$. According to \cite{4749273}, the ISL of a deterministic sequence $\mathbf{u}_n = \left[u_{1,n},\ldots,u_{N,n}\right]^T$ may be equivalently written as
    \begin{align}
        \nonumber&\sum\limits_{k = 1}^{N-1}|\mathbf{u}_n^H\mathbf{J}_k\mathbf{u}_n|^2 = \frac{1}{4N}\sum\limits_{p = 1}^{2N}\left(\left|\sum\limits_{q = 1}^{N}u_{q,n}e^{-\frac{j2\pi pq}{2N}}\right|^2-1\right)^2\\
        &=\frac{1}{4N}\sum\limits_{p = 1}^{2N}\left(\left|\sum\limits_{q = 1}^{N}u_{q,n}e^{-\frac{j2\pi pq}{2N}}\right|^4 - 2\left|\sum\limits_{q = 1}^{N}u_{q,n}e^{-\frac{j2\pi pq}{2N}}\right|^2 + 1\right),
    \end{align}
    where
    \begin{align}
        \nonumber&\sum\limits_{p = 1}^{2N}\left|\sum\limits_{q = 1}^{N}u_{q,n}e^{-\frac{j2\pi pq}{2N}}\right|^4 = 4N^2\left\|\tilde{\mathbf{F}}_{2N}\mathbf{u}_n\right\|_4^4,\\
        &\sum\limits_{p = 1}^{2N}\left|\sum\limits_{q = 1}^{N}u_{q,n}e^{-\frac{j2\pi pq}{2N}}\right|^2 = 2N\left\|\tilde{\mathbf{F}}_{2N}\mathbf{u}_n\right\|_2^2 = 2N.
    \end{align}
    Therefore
    \begin{align}\label{ISL_U}
        \sum\limits_{n=1}^{N}\sum\limits_{k = 1}^{N-1}|\mathbf{u}_n^H\mathbf{J}_k\mathbf{u}_n|^2 \nonumber&= \sum\limits_{n=1}^{N}\left(N\left\|\tilde{\mathbf{F}}_{2N}\mathbf{u}_n\right\|_4^4 - \frac{1}{2}\right) \\
        &=N\left(\left\|\tilde{\mathbf{F}}_{2N}\mathbf{U}\right\|_4^4 -\frac{1}{2}\right).
    \end{align}
    Substituting \eqref{ISL_U} into \eqref{direct_sum} yields \eqref{A-ACF_EISL_closed_form}, completing the proof.
    
\section{Proof of Theorem \ref{thm:local_opt}}\label{thm_3_proof}
Let us first recast the objective function as
\begin{equation}
f(\mathbf{V})=\left\|\tilde{\mathbf{F}}_{2N}\mathbf{F}_N^{H}\mathbf{V}^H\right\|_4^4,\quad \mathbf{V}\in \mathbb{U}(N),
\end{equation}
where we recall that $\mathbf{V} = \mathbf{U}^H\mathbf{F}_N^H$. Therefore, the problem is equivalent to proving that $\mathbf{V} = \mathbf{I}$ is a local maximum of $f(\mathbf{V})$. Since $f(\mathbf{V})$ is defined over the unitary group, it suffices to show that the function $f(\mathbf{V})$ has a zero gradient at $\mathbf{V}=\mathbf{I}$, and that $f(\mathbf{V})$ is geodesically concave at $\mathbf{V}=\mathbf{I}$. This can be expressed as that
\begin{subequations}
\begin{align}
\left.\frac{{\rm d}}{{\rm d}t}(f\circ \gamma_{\mathbf{I}})(t)\right|_{t=0} &= 0, \label{first_order_derivative}\\
\left.\frac{{\rm d}^2}{{\rm d}t^2}(f\circ \gamma_{\mathbf{I}})(t)\right|_{t=0}&\leq 0\label{second_order_derivative}
\end{align}
\end{subequations}
hold for all geodesics $\gamma_{\mathbf{I}}(t) = \exp(t\bm{\Gamma})$ intersecting at $\mathbf{I}$, with $\bm{\Gamma}$ being an element in the tangent space $\mathcal{T}_{\mathbf{I}}\mathbb{U}(N)$. 
\subsubsection{Computing the First- and Second-Order Derivatives}
To facilitate the analysis, observe that the $\ell_4$ norm is preserved under permutations. This allows us to rearrange the rows of $\tilde{\mathbf{F}}_{2N}$ such that
\begin{equation}
f(\mathbf{V})=\left\|\mathbf{P}_{o,e}\tilde{\mathbf{F}}_{2N}\mathbf{F}_N^{H}\mathbf{V}^H\right\|_4^4,
\end{equation}
where $\mathbf{P}_{o,e}$ is a permutation matrix which separates the odd rows from the even rows of $\tilde{\mathbf{F}}_{2N}$, namely we have
\begin{equation}
    \mathbf{P}_{o,e}\tilde{\mathbf{F}}_{2N} = \left[\tilde{\mathbf{F}}_{2N,o}; \tilde{\mathbf{F}}_{2N,e}\right],
\end{equation}
with $\tilde{\mathbf{F}}_{2N,o},\tilde{\mathbf{F}}_{2N,e}\in\mathbb{C}^{N\times N}$ containing the odd and even rows of $\tilde{\mathbf{F}}_{2N}$, respectively. By exploiting the structure of DFT matrices, we have
\begin{equation}
    \tilde{\mathbf{F}}_{2N,o} = \frac{1}{\sqrt{2}}\mathbf{F}_N, \quad\tilde{\mathbf{F}}_{2N,e} = \frac{1}{\sqrt{2}}\mathbf{F}_N\mathbf{D}_{\frac{1}{2}},
\end{equation}
where $\mathbf{D}_{\alpha}\in \mathbb{C}^{N \times N}$ is a diagonal matrix with its $n$-th diagonal entry being $e^{j\frac{2\pi(n-1)}{N} \cdot \alpha}$. For notational convenience, let us further define $\mathbf{C}_{\alpha}=\frac{1}{\sqrt{2}}\mathbf{F}_N\mathbf{D}_{\alpha}\mathbf{F}_N^{H}$, which is a circular matrix. As a consequence, we have
\begin{equation}
    \mathbf{P}_{o,e}\tilde{\mathbf{F}}_{2N}\mathbf{F}_N^H = \left[\frac{1}{\sqrt{2}}\mathbf{I};\frac{1}{\sqrt{2}}\mathbf{F}_N\mathbf{D}_{\frac{1}{2}}\mathbf{F}_N^H\right] = \left[\mathbf{C}_0;\mathbf{C}_{\frac{1}{2}}\right].
\end{equation}
Next, we note that the tangent space of $\mathbb{U}(N)$ is its Lie algebra, which is the set of all $N\times N$ skew Hermitian matrices. This implies that for each geodesic $\gamma_{\mathbf{I}}(t) = \exp(t\mathbf{\Gamma})$, we have $\mathbf{\Gamma}=j\mathbf{H}$, where $\mathbf{H}$ is a Hermitian matrix. We may then express the geodesic as a power series around $\mathbf{I}$, given by
\begin{equation}
    \gamma_{\mathbf{I}}(t) = \exp(jt\mathbf{H}) = \mathbf{I} + j t\mathbf{H} -\frac{t^2}{2}\mathbf{H}^2 + o(t^2).
\end{equation}
It follows that
\begin{align}
    &\nonumber(f\circ\gamma_{\mathbf{I}})(t) = \left\|\mathbf{P}_{o,e}\tilde{\mathbf{F}}_{2N}\mathbf{F}_N^{H}\gamma_{\mathbf{I}}(t)\right\|_4^4\\
    &\hspace{3mm}\nonumber= \left\|\mathbf{C}_0\left(\mathbf{I}+ jt\mathbf{H}-\frac{t^2}{2}\mathbf{H}^2\right)\right\|_4^4\\
    &\hspace{3mm}\quad\;\;+\left\|\mathbf{C}_{\frac{1}{2}}\left(\mathbf{I}+ jt\mathbf{H}-\frac{t^2}{2}\mathbf{H}^2\right)\right\|_4^4+o(t^2).
\end{align}
By representing the entry-wise square of a matrix as $\left|\mathbf{X}\right|^2 = \mathbf{X}\odot\mathbf{X}^\ast$, we have $\left\|\mathbf{X}\right\|_4^4 = {\rm{Tr}}\left\{{\left|\mathbf{X}\right|^2}^{T}\left|\mathbf{X}\right|^2\right\}$. Based on that, one may express
\begin{align}
&\nonumber\left\|\mathbf{C}_{\alpha}\left(\mathbf{I}+j\mathbf{H}-\frac{t^2}{2}\mathbf{H}^2\right)\right\|_4^4 \\
&\nonumber= {\rm Tr}\left\{{|\mathbf{C}_{\alpha}|^{2}}^T|\mathbf{C}_{\alpha}|^2\right\} + 4t{\rm ReTr}\left\{{|\mathbf{C}_{\alpha}|^{2}}^T{\rm Im}\left(\mathbf{C}_{\alpha}\mathbf{H}\odot \mathbf{C}_{\alpha}^\ast\right)\right\}\\
&+2t^2{\rm ReTr}\left\{{|\mathbf{C}_{\alpha}|^2}^T\left[|\mathbf{C}_{\alpha}\mathbf{H}|^2-{\rm Re}\left(\mathbf{C}_{\alpha}\mathbf{H}^2\odot\mathbf{C}_{\alpha}^\ast\right)\right]\right\}+o(t^2),
\end{align}
leading to
\begin{align}\label{first_order_expression}
    &\nonumber \left.\frac{{\rm d}}{{\rm d}t}(f\circ \gamma_{\mathbf{I}})(t)\right|_{t=0}\\
    & \hspace{3mm}\!=\! 4{\rm Tr}\left\{{|\mathbf{C}_{\frac{1}{2}}|^{2}}^T{\rm Im}\left(\mathbf{C}_{\frac{1}{2}}\mathbf{H}\!\odot\! \mathbf{C}_{\frac{1}{2}}^\ast\right)\!+\!{|\mathbf{C}_{0}|^{2}}^T{\rm Im}\left(\mathbf{C}_{0}\mathbf{H}\!\odot\! \mathbf{C}_{0}^\ast\right)\right\},
\end{align}
and
\begin{subequations}\label{second_order_expression}
\begin{align}
    &\nonumber\left.-\frac{{\rm d}^2}{{\rm d}t^2}(f\circ \gamma_{\mathbf{I}})(t)\right|_{t=0}\\
    &\hspace{3mm} = 4{\rm Tr}\left\{{|\mathbf{C}_{\frac{1}{2}}|^{2}}^T\left[{\rm Re}\left(\mathbf{C}_{\frac{1}{2}}\mathbf{H}^2\odot\mathbf{C}_{\frac{1}{2}}^\ast\right)-|\mathbf{C}_{\frac{1}{2}}\mathbf{H}|^2\right]\right\}\label{hessian_first_part}\\
    &\hspace{7mm}+4{\rm Tr}\left\{{|\mathbf{C}_{0}|^{2}}^T\left[{\rm Re}\left(\mathbf{C}_{0}\mathbf{H}^2\odot\mathbf{C}_{0}^\ast\right)-|\mathbf{C}_{0}\mathbf{H}|^2\right]\right\}\label{hessian_second_part}.
\end{align}
\end{subequations}
In the remainder of the proof, we show that $\eqref{first_order_expression} = 0$ and $\eqref{second_order_expression} \ge 0$ holds for any Hermitian matrix $\mathbf{H}$.
\subsubsection{First-Order Derivative is Zero}
Since $\mathbf{C}_0 = \frac{1}{\sqrt{2}}\mathbf{I}$, it is straightforward to observe that
\begin{align}
{\rm Tr}\left\{{|\mathbf{C}_{0}|^{2}}^T{\rm Im}\left(\mathbf{C}_{0}\mathbf{H}\odot\mathbf{C}_{0}^\ast\right)\right\} = \frac{1}{4}{\rm Tr}\left\{{\rm Im}\left(\mathbf{H}\right)\right\}=0,
\end{align}
due to the fact that the diagonal entries of a Hermitian matrix are all real. 

Next, we consider the remaining part of \eqref{first_order_expression}. Note that
\begin{align}
    &\nonumber{\rm Tr}\left\{{|\mathbf{C}_{\frac{1}{2}}|^{2}}^T{\rm Im}\left(\mathbf{C}_{\frac{1}{2}}\mathbf{H}\odot\mathbf{C}_{\frac{1}{2}}^\ast\right)\right\}\\
    &={\rm Tr}\left\{\mathbf{F}_N^{H}{|\mathbf{C}_{\frac{1}{2}}|^{2}}^T\mathbf{F}_N\mathbf{F}_N^{H}{\rm Im}\left(\mathbf{C}_{\frac{1}{2}}\mathbf{H}\odot\mathbf{C}_{\frac{1}{2}}^\ast\right)\mathbf{F}_N\right\},
\end{align}
which allows us to rewrite
\begin{align}
    &\nonumber \mathbf{F}_N^{H}{|\mathbf{C}_{\frac{1}{2}}|^{2}}^T\mathbf{F}\\
    &\hspace{3mm}\nonumber = \frac{1}{2}\mathbf{F}_N^{H} \left(\mathbf{F}_N\mathbf{D}_{\frac{1}{2}}\mathbf{F}_N^{H} \odot \mathbf{F}_N\mathbf{F}_N^{H}\mathbf{F}_N^\ast\mathbf{D}_{\frac{1}{2}}^\ast\mathbf{F}_N^{T}\mathbf{F}_N\mathbf{F}_N^{H}\right)\mathbf{F}_N\\
    &\hspace{3mm}=\frac{1}{2}\left(\mathbf{D}_{\frac{1}{2}}\circledast\mathbf{U}_{\rm TR}\mathbf{D}_{\frac{1}{2}}^\ast\mathbf{U}_{\rm TR}^{H}\right), \label{2d_conv_representation1}
\end{align}
where $\mathbf{U}_{\rm TR}=\mathbf{F}_N^{H}\mathbf{F}_N^\ast$ is a permutation matrix, and $\circledast$ denotes the circular convolution. For $N\times N$ matrices $\mathbf{A}$ and $\mathbf{B}$, the 2-D circular convolution is defined as
\begin{align}
    \mathbf{A}\circledast\mathbf{B} = \mathbf{F}_N^{H}\left(\mathbf{F}_N\mathbf{A}\mathbf{F}_N^{H}\odot\mathbf{F}_N\mathbf{B}\mathbf{F}_N^{H} \right)\mathbf{F}_N,
\end{align}
with its $(i,k)$-th entry being defined as
\begin{equation}
    [\mathbf{A}\circledast\mathbf{B}]_{i,k} = \sum_{m=1}^N\sum_{n=1}^N [\mathbf{A}]_{m,n} [\mathbf{B}]_{\overline{i-m}+1,\overline{k-n}+1},
\end{equation}
where $\overline{n} = n \mod N$ for integer $n$. Moreover, for length-$N$ vectors $\mathbf{a}$ and $\mathbf{b}$, the 1-D circular convolution is
\begin{equation}
    \mathbf{a}\circledast\mathbf{b} = \mathbf{F}_N^{ H}\left(\mathbf{F}_N\mathbf{a}\odot \mathbf{F}_N\mathbf{b}\right).
\end{equation}
For notational simplicity, we denote $\mathbf{D}_{\frac{1}{2},{\rm TR}}=\mathbf{U}_{\rm TR}\mathbf{D}_{\frac{1}{2}}^\ast\mathbf{U}_{\rm TR}^{H}$ and $\tilde{\mathbf{H}}=\mathbf{F}_N^{H}\mathbf{H}\mathbf{F}_N$, and hence we have
\begin{align}
    &\nonumber\mathbf{F}_N^{H}{\rm Im}\left(\mathbf{C}_{\frac{1}{2}}\mathbf{H}\odot\mathbf{C}_{\frac{1}{2}}^\ast\right)\mathbf{F}_N \\
    &= -\frac{j}{4}\left(\mathbf{D}_{\frac{1}{2}}\tilde{\mathbf{H}}\circledast\mathbf{D}_{\frac{1}{2},\rm{TR}}-\mathbf{D}_{\frac{1}{2},\rm{TR}}\tilde{\mathbf{H}}\circledast\mathbf{D}_{\frac{1}{2}}\right). \label{2d_conv_representation2}
\end{align}
Therefore, \eqref{first_order_expression} may be expressed as
\begin{align}
 \eqref{first_order_expression} &\!=\! - \frac{j}{2}{\text{Tr}}\left\{
  \left( {{{\mathbf{D}}_{\frac{1}{2}}} \circledast {{\mathbf{D}}_{\frac{1}{2},{\text{TR}}}}} \right)\right. \nonumber \\
  &\hspace{15mm}\odot\mathbf{I}\odot \left.\left( {{{\mathbf{D}}_{\frac{1}{2}}}\tilde {\mathbf{H}} \circledast {{\mathbf{D}}_{\frac{1}{2},{\text{TR}}}} \!-\! {{\mathbf{D}}_{\frac{1}{2},{\text{TR}}}}\tilde {\mathbf{H}} \circledast {{\mathbf{D}}_{\frac{1}{2}}}} \right)  \right\}.
\end{align}
To proceed, we prove the following lemma.
\begin{lemma}
For any diagonal matrix $\mathbf{D}$ and another arbitrary matrix $\mathbf{A}$, we have
\begin{equation}\label{lemma_diag}
\mathbf{I}\odot(\mathbf{A}\circledast\mathbf{D}) = (\mathbf{I}\odot\mathbf{A})\circledast\mathbf{D}.
\end{equation}
\end{lemma}
\begin{proof}
By denoting $\mathbf{F}_N = [\mathbf{f}_1,\mathbf{f}_2,\dotsc,\mathbf{f}_N]$, $\mathbf{D}={\rm diag}(d_1,\dotsc,d_N)$, it follows that
\begin{align}
&\nonumber[\mathbf{I}\odot(\mathbf{A}\circledast\mathbf{D})]_{i,i} = \mathbf{f}_i^{H}(\mathbf{F}_N\mathbf{A}\mathbf{F}_N^{H}\odot \mathbf{F}_N\mathbf{D}\mathbf{F}_N^{H})\mathbf{f}_i\\
&\nonumber=\sum_k d_k \mathbf{f}_i^{H}(\mathbf{F}_N\mathbf{A}\mathbf{F}_N^{H}\odot \mathbf{f}_k\mathbf{f}_k^{H})\mathbf{f}_i \\
&\nonumber=\sum_m \sum_n\sum_k d_k [\mathbf{A}]_{m,n}\mathbf{f}_i^{H}(\mathbf{f}_m\mathbf{f}_n^{H}\odot \mathbf{f}_k\mathbf{f}_k^{H})\mathbf{f}_i \\
&\nonumber=\sum_m \sum_n\sum_k d_k [\mathbf{A}]_{m,n}\mathbf{f}_i^{H}(\mathbf{f}_m\odot \mathbf{f}_k)(\mathbf{f}_n\odot\mathbf{f}_k)^{H}\mathbf{f}_i \\
&=\sum_m \sum_n\sum_k d_k [\mathbf{A}]_{m,n}\mathbf{f}_i^{H}\mathbf{f}_{\overline{m+k}}\mathbf{f}_{\overline{n+k}}^{H}\mathbf{f}_i.
\end{align}
Note that the term $\mathbf{f}_k^{H}\mathbf{f}_{\overline{m+k}}\mathbf{f}_{\overline{n+k}}^{H}\mathbf{f}_i\ne 0$ only when $m=n$ and $m+k=i$. Hence
\begin{align}
    \left[\mathbf{I}\odot(\mathbf{A}\circledast\mathbf{D})\right]_{i,i} &\nonumber = \sum_k \sum_m d_k a_{mm}\mathbf{f}_i^{H}(\mathbf{f}_m \mathbf{f}_m^{H}\odot \mathbf{f}_k\mathbf{f}_k^{H})\mathbf{f}_i\\
    &= \left[(\mathbf{I}\odot\mathbf{A})\circledast\mathbf{D}\right]_{i,i},
\end{align}
yielding \eqref{lemma_diag}.
\end{proof}

Upon relying on this lemma, we may obtain
\begin{align}
 \eqref{first_order_expression} &\!=\! - \frac{j}{2}{\text{Tr}}\left\{
  \left( {{{\mathbf{D}}_{\frac{1}{2}}} \circledast {{\mathbf{D}}_{\frac{1}{2},{\text{TR}}}}} \right)\right. \nonumber \\
  &\hspace{3mm}\odot \left.\left( {\mathbf{D}_{\frac{1}{2}}(\mathbf{I}\odot\tilde{\mathbf{H}})\circledast\mathbf{D}_{\frac{1}{2},{\rm TR}} \!-\! \mathbf{D}_{\frac{1}{2},{\rm TR}}(\mathbf{I}\odot\tilde{\mathbf{H}})\circledast\mathbf{D}_{\frac{1}{2}}} \right)  \right\}.
\end{align}
By letting
\begin{align}
    &\nonumber\mathbf{d}_{\frac{1}{2}} = {\rm ddiag}(\mathbf{D}_{\frac{1}{2}}),\\
    &\nonumber\mathbf{d}_{\frac{1}{2},{\rm TR}} = {\rm ddiag}(\mathbf{D}_{\frac{1}{2},{\rm TR}}),\\
    & \mathbf{d}_h = {\rm ddiag}(\mathbf{I}\odot\tilde{\mathbf{H}}),\label{vec_notations}
\end{align}
we may represent \eqref{first_order_expression} as $-\frac{j}{2}c$, where
\begin{align}
    c &\!=\! \mathbf{1}^{T}\left\{(\mathbf{d}_{\frac{1}{2}}\circledast\mathbf{d}_{\frac{1}{2},{\rm TR}})\right. \nonumber \\
    &\hspace{10mm}\odot \left.\left[(\mathbf{d}_{\frac{1}{2}}\odot \mathbf{d}_h)\circledast\mathbf{d}_{\frac{1}{2},{\rm TR}} \!-\! (\mathbf{d}_{\frac{1}{2},{\rm TR}}\odot\mathbf{d}_h)\circledast \mathbf{d}_{\frac{1}{2}}\right]\right\}.
\end{align}
Now it suffices to show that $c$ is zero. To this end, let us ponder on the fact that
\begin{equation}
    [\mathbf{d}_{\frac{1}{2}}]_n = \frac{1}{\sqrt{N}}e^{-\frac{j\pi (n-1)}{N}},
\end{equation}
which implies that
\begin{align}
&\nonumber[(\mathbf{d}_{\frac{1}{2}}\circledast\mathbf{d}_{\frac{1}{2},{\rm TR}})\odot((\mathbf{d}_{\frac{1}{2}}\odot \mathbf{d}_h)\circledast\mathbf{d}_{\frac{1}{2},{\rm TR}})]_n \\
&\nonumber= \sum_{k} [\mathbf{d}_{\frac{1}{2}}]_k[\mathbf{d}_{\frac{1}{2}}]_{\overline{k-n}+1}^\ast \sum_l [\mathbf{d}_{\frac{1}{2}}]_l[\mathbf{d}_h]_l [\mathbf{d}_{\frac{1}{2}}]_{\overline{l-n}+1}^\ast\\
&\nonumber=\frac{1}{N^2}\sum_k \sum_l [\mathbf{d}_h]_l\cdot e^{-\frac{j\pi}{N}(l-k+\overline{k-n}-\overline{l-n})}\\
&\nonumber=\frac{1}{N^2}\sum_k e^{\frac{j\pi}{N}(k-\overline{k-n})}\sum_l [\mathbf{d}_h]_l e^{-\frac{j\pi}{N}(l-\overline{l-n})}\\
&\nonumber=\frac{1}{N^2}\left|e^{-\frac{j\pi}{N}(1-\overline{1-n})}\right|^2 (N-2(n-1))\left(\sum_{l=n}^{N}[\mathbf{d}_h]_l-\sum_{l=1}^{n-1}[\mathbf{d}_h]_l\right)\\
&=\frac{N-2(n-1)}{N^2}\left(\sum_{l=n}^{N}[\mathbf{d}_h]_l-\sum_{l=1}^{n-1}[\mathbf{d}_h]_l\right),
\end{align}
Similar arguments can also be applied to $(\mathbf{d}_{\frac{1}{2}}\circledast\mathbf{d}_{\frac{1}{2},{\rm TR}})\odot\left[(\mathbf{d}_{\frac{1}{2},{\rm TR}}\odot\mathbf{d}_h)\circledast \mathbf{d}_{\frac{1}{2}}\right]$, yielding
\begin{align}
&\nonumber[(\mathbf{d}_{\frac{1}{2}}\circledast\mathbf{d}_{\frac{1}{2},{\rm TR}})\odot((\mathbf{d}_{\frac{1}{2},{\rm TR}}\odot\mathbf{d}_h)\circledast \mathbf{d}_{\frac{1}{2}})]_n\\
&= \frac{N-2(n-1)}{N^2}\left(\sum_{l=1}^{N-n+1}[\mathbf{d}_h]_l - \sum_{l=N-n+2}^N [\mathbf{d}_h]_l\right).
\end{align}
Consequently, we have
\begin{align}
    c &\!=\! \sum_n c_n \triangleq \sum_n\frac{N-2(n-1)}{N^2}\left(\sum_{l=n}^{N}[\mathbf{d}_h]_l+\sum_{l=N-n+2}^N [\mathbf{d}_h]_l \right. \nonumber \\ 
    &\hspace{40mm} - \left.\sum_{l=1}^{n-1}[\mathbf{d}_h]_l \!-\! \sum_{l=1}^{N-n+1}[\mathbf{d}_h]_l\right).
\end{align}
Without the loss of generality, let us assume that $N$ is an even number, while similar argument also holds for the case that $N$ is odd. It is not difficult to see that 
\begin{align}
    c_1 = \frac{1}{N}\left(\sum_{l=1}^{N}[\mathbf{d}_h]_l - \sum_{l=1}^{N}[\mathbf{d}_h]_l\right) = 0,\quad c_{\frac{N+2}{2}} = 0.
\end{align}
Moreover, for any $1<n<\frac{N+2}{2}$, we have
\begin{align}
    & c_{N-n+2}\!=\! - \frac{N-2(n-1)}{N^2}\left(\sum_{l=N-n+2}^{N}[\mathbf{d}_h]_l+\sum_{l=n}^N [\mathbf{d}_h]_l \right. \nonumber \\ 
    &\hspace{30mm} - \left.\sum_{l=1}^{N-n+1}[\mathbf{d}_h]_l \!-\! \sum_{l=1}^{n-1}[\mathbf{d}_h]_l\right) = -c_n,
\end{align}
which suggests that $c = \sum_n c_n = 0$, and thus $\eqref{first_order_expression} = 0$.


\subsubsection{Second-Order Derivative is Non-Positive}
Next, let us compute \eqref{hessian_first_part} and \eqref{hessian_second_part}. By leveraging the 2-D circular convolution representation methods as in \eqref{2d_conv_representation1} and \eqref{2d_conv_representation2}, and following the notations in \eqref{vec_notations}, one may simplify \eqref{hessian_first_part} as
\begin{align}
    \eqref{hessian_first_part}&\nonumber=\mathbf{1}^{T}\left\{(\mathbf{d}_{\frac{1}{2}}\circledast\mathbf{d}_{\frac{1}{2},{\rm TR}}) \right.\\
    &\nonumber\odot\left.\left(\frac{1}{2}\left[(\mathbf{d}_{\frac{1}{2}}\odot\mathbf{d}_{h^2})\circledast\mathbf{d}_{\frac{1}{2},{\rm TR}}+(\mathbf{d}_{\frac{1}{2},{\rm TR}}\odot\mathbf{d}_{h^2})\circledast\mathbf{d}_{\frac{1}{2}}\right]\right.\right.\\
    &\hspace{7mm}\left.\left.-{\rm ddiag}\left[\mathbf{D}_{\frac{1}{2}}\tilde{\mathbf{H}}\circledast\mathbf{U}_{\rm TR}(\mathbf{D}_{\frac{1}{2}}\tilde{\mathbf{H}})^\ast\mathbf{U}_{\rm TR}^{\rm H}\right]\right)\right\},
\end{align}
where $\mathbf{d}_{h^2} = {\rm ddiag}\left(\mathbf{I}\odot\tilde{\mathbf{H}}^2\right)$. We note that
\begin{align}
&\nonumber\left[(\mathbf{d}_{\frac{1}{2}}\circledast\mathbf{d}_{\frac{1}{2},{\rm TR}}) \odot{\rm ddiag}(\mathbf{D}_{\frac{1}{2}}\tilde{\mathbf{H}}\circledast\mathbf{U}_{\rm TR}(\mathbf{D}_{\frac{1}{2}}\tilde{\mathbf{H}})^\ast\mathbf{U}_{\rm TR}^{\rm H})\right]_i\\
&\nonumber=\sum_k[\mathbf{d}_{\frac{1}{2}}]_k[\mathbf{d}_{\frac{1}{2}}]_{\overline{k-i}+1}^\ast\sum_m [\mathbf{d}_{\frac{1}{2}}]_m[\mathbf{d}_{\frac{1}{2}}]_{\overline{m-i}+1}^\ast\\
&\hspace{5mm}\cdot\sum_n [\tilde{\mathbf{H}}]_{m,n}[\tilde{\mathbf{H}}]_{\overline{m-i}+1,\overline{n-i}+1}^\ast.\label{hessian_first_part_1}
\end{align}
To proceed, we construct a matrix $\mathbf{G}$, such that
\begin{equation}
    \left[\mathbf{G}\right]_{m,i} = \sum_n [\tilde{\mathbf{H}}]_{m,n}[\tilde{\mathbf{H}}]_{\overline{m-i}+1,\overline{n-i}+1}^\ast,
\end{equation}
with which we have
\begin{equation}
    \eqref{hessian_first_part_1} =\frac{N-2(i-1)}{N^2}\left(\sum_{m=i}^{N}[\mathbf{G}]_{m,i}-\sum_{m=1}^{i-1}[\mathbf{G}]_{m,i}\right).
\end{equation}
This implies that
\begin{align}
&\nonumber\mathbf{1}^T\left[(\mathbf{d}_{\frac{1}{2}}\circledast\mathbf{d}_{\frac{1}{2},{\rm TR}}) \odot{\rm ddiag}(\mathbf{D}_{\frac{1}{2}}\tilde{\mathbf{H}}\circledast\mathbf{U}_{\rm TR}(\mathbf{D}_{\frac{1}{2}}\tilde{\mathbf{H}})^\ast\mathbf{U}_{\rm TR}^{\rm H})\right] \\
&= \frac{1}{N}{\rm Tr}\left(\mathbf{Q}^T\mathbf{G}\right),
\end{align}
where the $i$-th column of $\mathbf{Q}\in\mathbb{R}^{N \times N}$ is defined as
\begin{equation}
    \mathbf{q}_i = \frac{N-2(i-1)}{N}\left[-\mathbf{1}_{i-1};\mathbf{1}_{N-i+1}\right].
\end{equation}
By using the same technique, we further obtain
\begin{align}
    &\nonumber\mathbf{1}^{T}\left((\mathbf{d}_{\frac{1}{2}}\circledast\mathbf{d}_{\frac{1}{2},{\rm TR}})\right.\\
    &\nonumber\hspace{5mm}\left.\odot\frac{1}{2}\left[(\mathbf{d}_{\frac{1}{2}}\odot\mathbf{d}_{h^2})\circledast\mathbf{d}_{\frac{1}{2},{\rm TR}}+(\mathbf{d}_{\frac{1}{2},{\rm TR}}\odot\mathbf{d}_{h^2})\circledast\mathbf{d}_{\frac{1}{2}}\right]\right)\\
    & \hspace{5mm}= \frac{1}{N}{\rm Tr}\left(\mathbf{Q}^T\mathbf{d}_{h^2}\mathbf{1}^T\right).
\end{align}
By defining $\overline{\mathbf{H}} = \left[\overline{\mathbf{h}}_1,\overline{\mathbf{h}}_2,\ldots,\overline{\mathbf{h}}_N\right]$, with its $\left(n,m\right)$-th entry being $\overline{h}_{n,m} = [\tilde{\mathbf{H}}]_{\overline{n-m}+1,m}$, it follows that
\begin{equation}
\left[\mathbf{G}\right]_{m,i} = \overline{\mathbf{h}}_{\overline{m-i}+1}^{H}\overline{\mathbf{h}}_m,
\end{equation}
This indicates that upon denoting
\begin{equation}
[\overline{\mathbf{Q}}]_{m,i} = [\mathbf{Q}]_{m,\overline{N-i+m}+1},
\end{equation}
we have
\begin{equation}
{\rm Tr}\left(\mathbf{Q}^T\mathbf{G}\right)= {\rm Tr}\left(\overline{\mathbf{H}}^{H}\overline{\mathbf{Q}}^T\overline{\mathbf{H}}\right),
\end{equation}
which results in
\begin{align}
    \eqref{hessian_first_part} &\nonumber = \frac{1}{N}{\rm Tr}\left(\mathbf{Q}^T\mathbf{d}_{h^2}\mathbf{1}^{T}-\mathbf{Q}^T\mathbf{G}\right)\\
    & =\frac{1}{N}{\rm Tr}\left\{\overline{\mathbf{H}}^{H}\left[{\rm diag}\left(\overline{\mathbf{Q}}\mathbf{1}\right)-\overline{\mathbf{Q}}\right]\overline{\mathbf{H}}\right\}.
\end{align}
Following a similar procedure, and by replacing $\mathbf{D}_{\frac{1}{2}}$ with $\mathbf{D}_0$, \eqref{hessian_second_part} can be represented as
\begin{equation}
    \eqref{hessian_second_part} = \frac{1}{N}{\rm Tr}\left\{\overline{\mathbf{H}}^{H}\left(\mathbf{I}-\mathbf{1}\mathbf{1}^T\right)\overline{\mathbf{H}}\right\},
\end{equation}
yielding
\begin{align}
    &\nonumber\eqref{hessian_first_part}+\eqref{hessian_second_part} = \\
    &\frac{1}{N}{\rm Tr}\left\{\overline{\mathbf{H}}^{H}\left[{\rm diag}\left(\overline{\mathbf{Q}}\mathbf{1}\right)+\mathbf{I}-\left(\overline{\mathbf{Q}}+\mathbf{1}\mathbf{1}^T\right)\right]\overline{\mathbf{H}}\right\}. \label{hessian_quadratic_form}
\end{align}
Our final task is to prove that ${\rm diag}\left(\overline{\mathbf{Q}}\mathbf{1}\right)+\mathbf{I}-\left(\overline{\mathbf{Q}}+\mathbf{1}\mathbf{1}^T\right)$ is positive semi-definite. It is not difficult to see that $\overline{\mathbf{Q}}$ is symmetric. Therefore, it suffices to show that ${\rm diag}\left(\overline{\mathbf{Q}}\mathbf{1}\right)+\mathbf{I}-\left(\overline{\mathbf{Q}}+\mathbf{1}\mathbf{1}^T\right)$ is diagonal dominant. Observe that
\begin{equation}
    {\rm diag}\left(\overline{\mathbf{Q}}\mathbf{1}\right)+\mathbf{I} = {\rm diag}\left\{\mathbf{1}^T\left(\overline{\mathbf{Q}} + \mathbf{1}\mathbf{1}^T\right)\right\}.
\end{equation}
Furthermore, each entry of $\overline{\mathbf{Q}}$ is not smaller than $-1$ by its definition. This implies that the entries of $\overline{\mathbf{Q}} + \mathbf{1}\mathbf{1}^T$ are all non-negative, and hence the absolute values of these entries are themselves. This leads to
\begin{equation}
    \left[{\rm diag}\left(\overline{\mathbf{Q}}\mathbf{1}\right)+\mathbf{I}\right]_{i,i} - \sum_k\left|\left[\overline{\mathbf{Q}}+\mathbf{1}\mathbf{1}^T\right]_{i,k}\right| = 0, \forall i,
\end{equation}
suggesting that the matrix ${\rm diag}\left(\overline{\mathbf{Q}}\mathbf{1}\right)+\mathbf{I}-\left(\overline{\mathbf{Q}}+\mathbf{1}\mathbf{1}^T\right)$ is diagonal dominant, and hence the quadratic form in \eqref{hessian_quadratic_form} is non-negative for any $\overline{\mathbf{H}}$. This indicates that the second-order derivative is non-positive, thereby $\mathbf{V} = \mathbf{I}$ is a local maximum of $f\left(\mathbf{V}\right)$, concluding the proof.

	\bibliographystyle{IEEEtran}
	\bibliography{IEEEabrv,references_SPM,references,database}

\end{document}